\documentclass[conference]{IEEEtran}

\IEEEoverridecommandlockouts
\usepackage[T1]{fontenc}
\usepackage{cite}
\usepackage{amsmath,amssymb,amsfonts}
\usepackage{algorithmic}
\usepackage{graphicx}
\usepackage{textcomp}
\usepackage{xcolor}
\def\BibTeX{{\rm B\kern-.05em{\sc i\kern-.025em b}\kern-.08em
    T\kern-.1667em\lower.7ex\hbox{E}\kern-.125emX}}

\usepackage{custom_packages}
\def\T{{\mathsf T}} 
\def\tr{\mathrm{tr}} 

\def\F{\mathbf{F}} 
\def\A{\mathbf{A}} 
\def\I{\mathbf{I}} 
\def\M{\mathbf{M}} 
\def\X{\mathbf{X}} 
\def\S{\mathbf{S}} 
\def\O{\mathbf{O}} 
\def\B{\mathbf{B}} 
\def\Cov{\boldsymbol{\Sigma}} 

\def\g{\mathbf{g}} 
\def\q{\mathbf{q}} 
\def\a{\mathbf{a}} 
\def\b{\mathbf{b}} 
\def\v{\mathbf{v}} 
\def\u{\mathbf{u}} 
\def\x{\mathbf{x}} 
\def\y{\mathbf{y}} 
\def\n{\mathbf{n}} 
\def\z{\mathbf{z}} 
\def\w{\mathbf{w}} 
\def\h{\mathbf{h}} 
\def\f{\mathbf{f}} 
\def\u{\mathbf{u}} 

\def\Ptx{P_{\rm fw}}
\def\Prx{P_{\rm fb}}
\def\ctx{c_{\rm fw}}
\def\crx{c_{\rm fb}}
\def\SNR{\mathrm{SNR}} 
\def\MMSE{\mathrm{MMSE}} 
\def\Gap{\mathrm{Gap}} 

\def\fw{\mathrm{fw}} 
\def\fb{\mathrm{fb}} 

\begin{document}
\title{
KKTCode: Asymptotically Optimal Linear Codes for Noisy Feedback AWGN Channels

\thanks{H. Nam, V. Tripathi, and D. J. Love are with the Elmore Family School of Electrical and Computer Engineering, Purdue University, West Lafayette, IN 47907 USA (e-mail: \{nam86,  tripathv, djlove\}@purdue.edu).
J. Jang is with the Department of Electronics Engineering, Chungnam National University, Daejeon 34134, Republic of Korea (e-mail: jgjang@cnu.ac.kr).
This work was supported in part by the Office of Naval Research (ONR) under Grant N000142112472, and the National Science Foundation (NSF) under Grants CNS2212565, CNS2225578, and EEC1941529. }
}

\author{
Hongjae Nam,~\IEEEmembership{Student Member,~IEEE,}
Jonggyu Jang,~\IEEEmembership{Member,~IEEE,}\\ 
Vishrant Tripathi,~\IEEEmembership{Member,~IEEE,} 
and David J. Love,~\IEEEmembership{Fellow,~IEEE}
}

\maketitle

\begin{abstract}
The design of optimal causal linear feedback schemes for additive white Gaussian noise (AWGN) channels with noisy output feedback has remained an open problem for over 60 years. Prior work has focused on restricted policy classes, especially passive (uncoded) noisy output feedback, where only the transmitter performs feedback coding. However, passive noisy output feedback fundamentally lacks the degrees of freedom required to attain the information-theoretic performance limit in general. In this paper, we consider the active (coded) noisy output feedback setting, where both the transmitter and the receiver perform feedback coding. We then develop a constructive KKT-optimal active linear feedback design that asymptotically attains the Elias--Butman SNR converse bound, thereby establishing MSE/SNR optimality over the entire class of causal linear schemes. Furthermore, we prove that the optimal passive feedback solution is recovered as a special case of the active design. This passive solution admits a Geometric Toeplitz (GT) structure with a Chance--Love (CL)-style one-shot polynomial characterization, and can be computed with $\mathcal{O}(\log T)$ complexity. 
Thus, our results provide an affirmative answer to the long-standing optimality question for noisy output feedback under causal linear feedback coding, and our numerical results support the theoretical findings.
\end{abstract}

\begin{IEEEkeywords}
Additive Gaussian channels, noisy feedback, linear codes.
\end{IEEEkeywords}

\section{Introduction}
\label{sec:Intro}

\IEEEPARstart{F}{eedback} plays a fundamental role in modern wireless communication systems. Early results by Shannon~\cite{shannon1956zero} showed that, for memoryless channels with noiseless feedback, feedback does not increase capacity in the asymptotic regime. However, a significant body of work over the last seven decades~\cite{elias1957channel, schalkwijk1966coding1 ,schalkwijk1966coding2, elias1967networks, butman1969general, bacsar1982optimum, chance2011concatenated, farthofer2014achieving, mishra2023linear, agrawal2011iteratively, kim2023robust, li2011bounds, kim2006feedback, kim2009feedback, kim2011error, domanovitz2024information} has shown that feedback can dramatically influence performance in the finite-blocklength regime. 


For additive white Gaussian noise (AWGN) channels with noisy output feedback, this leads to a more challenging problem. Unlike the noiseless feedback setting, the encoder, the feedback operations, and the final estimator must be jointly designed in the presence of feedback noise and separate power constraints on the forward and feedback links. As a result, even within the class of causal linear schemes, the optimal finite-blocklength feedback scheme and its structure remain unknown. Moreover, existing works~\cite{schalkwijk1966coding1,  schalkwijk1966coding2, bacsar1982optimum, chance2011concatenated, agrawal2011iteratively ,mishra2023linear, kim2023robust} have primarily focused on the \textit{passive} (uncoded) noisy output feedback setting, where only the transmitter performs feedback coding. This restricted subclass limits the degrees of freedom available for shaping the effective noise and improving the received signal-to-noise ratio (SNR). By contrast, the more general \textit{active} (coded) noisy output feedback setting allows both the transmitter and the receiver to perform feedback coding and remains much less understood.

In this work, we address the following central question:
\begin{mdframed}
\textit{For any finite blocklength $T$, how do we design \textbf{optimal causal linear feedback schemes} for AWGN channels with \textbf{active} noisy output feedback?} 
\end{mdframed}
From a theoretical viewpoint, we are interested in whether the received SNR converse bound (Elias--Butman) \cite{elias1967networks,butman1969general} is tight over the class of causal linear feedback schemes. 
From a practical viewpoint, we want to move beyond heuristic and restricted constructions and develop a computationally tractable design approach for coding in feedback channels.

Our main result is an affirmative answer to the above question: we develop a constructive KKT-optimal active linear feedback design that attains the Elias--Butman SNR converse bound, thereby establishing MSE/SNR optimality of the general causal linear class. We further show that the globally optimal passive output-feedback solution arises as a special case of our active formulation and admits a geometric-Toeplitz (GT) structure.

\paragraph*{\textbf{Related Work}} 

The search for the optimal linear feedback schemes for AWGN channels dates back to Elias' 1957 work on feedback channel capacity~\cite{elias1957channel} and the later Schalkwijk--Kailath (SK) scheme~\cite{schalkwijk1966coding1, schalkwijk1966coding2}, while Butman (1969) formulated the general noisy-output-feedback linear design problem and derived the Elias--Butman SNR converse (information-theoretic upper bound). However, whether this converse is achievable has remained open due to the difficulty of solving the resulting coupled nonlinear equations~\cite{elias1967networks, butman1969general}. 
Decades later, the Chance--Love (CL) scheme proposed in~\cite{chance2011concatenated} extended this work to practical AWGN channels with \emph{noisy} output feedback. Related work considered CL-style linear processing with quantized, rate-constrained feedback~\cite{farthofer2014achieving}. However, the CL scheme is restricted to passive feedback and relies on asymptotic assumptions and structural conjectures. 
Most recently, Mishra \emph{et al.}~\cite{mishra2023linear} proposed a dynamic-programming (DP)-based scheme for the passive-feedback case by casting the noisy-feedback coding problem as a decentralized linear-quadratic-Gaussian (LQG) control problem. However, the DP scheme is recursive and MSE-optimal only within the restricted class of scalar-state sequential linear schemes, leaving optimality over the full causal linear class open. 

Beyond linear feedback, nonlinear schemes such as Modulo-SK and learning-based feedback codes such as GBAF, DeepCode, and LightCode have also been studied~\cite{ben2015gaussian, ben2017interactive, kim2018deepcode, jiang2019turbo, jiang2020feedback, ozfatura2022all, ozfatura2023feedback, chahine2022inventing, ankireddy2025lightcode, ben2020simple_Modulo-SK_vs_Deepcode, kim2020deepcode_vs_Modulo-SK, kim2023robust, jang2025syndromecode, ding2026deep}. In particular, learning-based active-feedback schemes such as BAAF~\cite{ozfatura2023feedback} and ActiveFB~\cite{chahine2022inventing} jointly design the forward and feedback coding operations, although their theoretical optimality remains unclear.
Related two-way feedback studies have considered both linear Gaussian formulations based on SNR-constrained power optimization~\cite{kim2022linear} and learning-based schemes evaluated mainly by empirical block-error-rate (BLER) performance~\cite{kim2025coding, nickel2025learning}.

A comparison of the key properties of the proposed SNR/MSE-optimal scheme and other linear schemes discussed above is summarized in Table~\ref{tab:scheme-comparison}.

\begin{table}[t]
\centering
\setlength{\tabcolsep}{6pt}
\renewcommand{\arraystretch}{1.2}
\begin{tabular}{lcccc}
\textbf{Scheme} & 
\makecell{\textbf{Noisy}\\ \textbf{feedback}} & 
\makecell{\textbf{SNR}\\ \textbf{optimal}}  &
\makecell{\textbf{MSE}\\ \textbf{optimal}}  & 
\makecell{\textbf{Active}\\ \textbf{feedback}}\\
\hline
\rowcolor[gray]{0.9}
KKT (proposed)                   & O & O & O & O \\
SK~\cite{schalkwijk1966coding1} (1966)  & X & X & X & X \\
CL~\cite{chance2011concatenated} (2011) & O & O & X & X  \\
DP~\cite{mishra2023linear} (2023)       & O & O$^\ast$ & O$^\ast$ & X  \\
\end{tabular}

\vspace{1pt}
{\footnotesize $^\ast$Optimal only within the restricted scalar sequential linear schemes.}

\caption{Comparison of linear feedback schemes.}
\label{tab:scheme-comparison}
\end{table}

\paragraph*{\textbf{Contributions}}
In this work, we study the full class of causal linear feedback schemes with noisy output feedback. 
Our key contributions are:
\begin{compactenum}    
  \item We develop a constructive KKT-optimal active feedback design for every finite blocklength. The proposed scheme asymptotically achieves the Elias--Butman SNR converse, thereby establishing MSE/SNR optimality in the general causal linear feedback class.

  \item  We show that the passive specialization of the KKT-based design is globally optimal within the passive-feedback class and admits a geometric-Toeplitz (GT) structure. 

  \item Unlike prior passive schemes such as CL~\cite{chance2011concatenated} and DP~\cite{mishra2023linear}, the proposed KKT-based passive scheme is globally optimal for every finite blocklength $T$, non-recursive, and computable via low-complexity one-dimensional root-finding with $O(\log T)$ operations. 
    


\end{compactenum}

\paragraph*{Organization} The remainder of the paper is organized as follows. Section~\ref{sec:System model} introduces the linear AWGN noisy-feedback model. Section~\ref{sec:KKT-Optimization} develops a constructive KKT-based optimization for active feedback, including an eigenvector characterization, an envelope-gradient ascent algorithm, convergence analysis, and the globally optimal passive-feedback specialization. Section~\ref{sec:Achievability} presents converse bounds on SNR and MSE for causal linear feedback schemes and establishes the asymptotic achievability of the proposed design. In Section~\ref{sec:Numerical results}, we provide simulation results to validate our analysis and compare the proposed schemes against existing linear baselines, including the SK, CL and DP schemes, followed by the conclusion in Section~\ref{sec:Conclusion}.

\paragraph*{Notations}
 Boldface lowercase and uppercase letters denote vectors and matrices, respectively (e.g., $\mathbf{x}$ and $\mathbf{X}$).
 For a vector $\mathbf{x}$, $\|\mathbf{x}\|$ denotes the Euclidean ($\ell_2$) norm, and for a matrix $\mathbf{X}$, $\|\mathbf{X}\|_F$ denotes the Frobenius norm.
 For a Hermitian matrix $\mathbf{X}$, $\lambda_{\max}(\mathbf{X})$ and $\lambda_{\min}(\mathbf{X})$ denote its largest and smallest eigenvalues, and $\mathbf{v}_{\max}(\mathbf{X})$ and $\mathbf{v}_{\min}(\mathbf{X})$ denote corresponding unit-norm eigenvectors. For a general matrix $\mathbf{X}$, $\sigma_{\max}(\mathbf{X})$ and $\sigma_{\min}(\mathbf{X})$ denote its largest and smallest singular values. We use $\mathbf{I}$ for the identity matrix (dimension clear from context), and $\mathbf{X}\succeq \mathbf{0}$ to indicate that $\mathbf{X}$ is positive semidefinite. 

\section{System Model} \label{sec:System model}
\vspace{-3mm}
\begin{figure} [h!]
    \centering
    \includegraphics[width=1.0\linewidth]{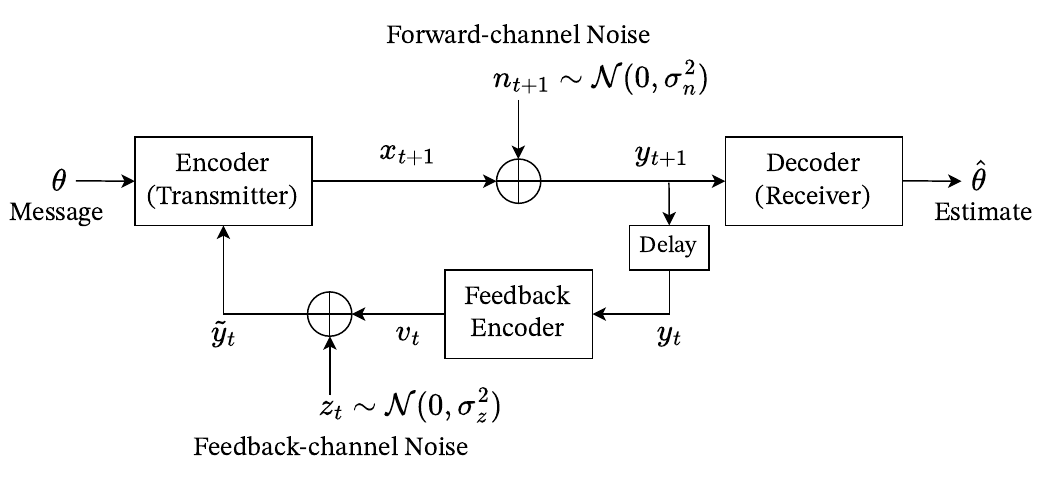}
    \caption{AWGN channel with noisy output feedback}
    \label{fig:active}
\end{figure}

We consider causal linear feedback schemes for transmitting a Gaussian source symbol
\(
    \theta \sim \mathcal N(0,P_\theta)
\)
over an additive white Gaussian noise (AWGN) channel with noisy output feedback, as illustrated in Fig.~\ref{fig:active}. Over $T$ channel uses, the forward and feedback links are 
\begin{align}
    y_t &= x_t + n_t,        \quad n_t \sim \mathcal N(0,\sigma_{n}^2), 
    \label{eq:active-forward-channel}\\
    \tilde y_t &= v_t + z_t, ~\quad z_t \sim \mathcal N(0,\sigma_{z}^2),         
    \label{eq:active-feedback-channel}
\end{align}
for $t=0,\dots,T-1$, where $x_t$ and $v_t$ denote the forward- and feedback-link channel inputs respectively, and $\{n_t\}$ and $\{z_t\}$ are mutually independent i.i.d. noise sequences that are also independent of $\theta$. Both the transmitter and receiver are subject to the average per-symbol power constraints:
\begin{align}
    \sum_{t=0}^{T-1} \mathbb{E}[x_t^2] \le T \Ptx, \qquad \sum_{t=0}^{T-1} \mathbb{E}[v_t^2] \le T \Prx,
    \label{eq:active-power-constraints}
\end{align} 
where $\Ptx>0$ and $\Prx>0$ are the forward- and feedback-link per-symbol power budgets, respectively.

For blockwise notation, define the following vectors:
\begin{align*}
    \mathbf{x} = [x_0,\ldots,x_{T-1}]^\T, \qquad
    \mathbf{y} = [y_0,\ldots,y_{T-1}]^\T,\\
    \mathbf{\tilde{y}} = [\tilde{y}_0,\ldots,\tilde{y}_{T-1}]^\T, \qquad
    \mathbf{v} = [v_0,\ldots,v_{T-1}]^\T, \\
    \mathbf{n} = [n_0,\ldots,n_{T-1}]^\T, \qquad
    \mathbf{z} = [z_0,\ldots,z_{T-1}]^\T. 
\end{align*}

A causal linear active-feedback scheme is written as:
\begin{align}
&\mathbf x=\mathbf g\,\theta+ \F\A \mathbf{n} + \F \mathbf{z}, \\
&\mathbf v=\mathbf A\,\mathbf y,
\label{eq:active-linear-law}
\end{align}
where $\mathbf g\in\mathbb R^{T}$ is the precoding vector, $\mathbf F\in\mathbb R^{T\times T}$  is the transmitter feedback coding matrix (transmitter filter), and $\mathbf A\in\mathbb R^{T\times T}$ is the receiver feedback coding matrix (receiver filter). Causality requires $\F$ to be strictly lower triangular and $\A$ to be lower triangular.

Equivalently, the received vector can be written as:
\begin{equation}
    \mathbf{y} =  \underbrace{\g \theta}_{\text{signal}} + \underbrace{\w}_{\substack{\text{effective}\\\text{noise}}}, \quad \w \triangleq (\I+\F\A)\n+\F\z,
\end{equation}
where $\w$ denotes the effective noise vector with covariance 
\begin{equation}
\Cov_w(\F,\A) \triangleq \sigma_n^2\big(\mathbf I+\mathbf F\mathbf A\big)\big(\mathbf I+\mathbf F\mathbf A\big)^{\mathsf T} + \sigma_z^2\mathbf{F}\mathbf{F}^{\mathsf T}.
\end{equation}
When the dependence on $(\F,\A)$ is clear from context, we simply write
$\Cov_w$ for $\Cov_w(\F,\A)$.

For fixed $(\g, \F, \A)$, the model is linear Gaussian, so the optimal decoder is the LMMSE estimator $\q \in \mathbb{R}^T$. After observing the received vector $\y$, the receiver forms the estimate
\begin{equation}
    \hat{\theta} = \mathbf{q}^\T \mathbf{y}, \quad \mathbf{q} \triangleq \frac{\Cov_w^{-1} \mathbf{g}}{1/P_\theta+ \mathbf{g}^\T \Cov_w^{-1} \mathbf{g}}.
    \label{eq:active-q-estimate}
\end{equation}
Thus, $\q$ is determined by $(\g, \F, \A)$, and the design reduces to the following SNR maximization problem over $(\g, \F, \A)$.

For a linear estimator $\hat{\theta}=\q^\T\y$, we define the
received SNR as the post-processed signal-to-noise ratio
\begin{equation}
    \SNR(\q;\g,\F,\A)
    \triangleq \frac{\mathbb{E}[(\q^\T\g\theta)^2]}
    {\mathbb{E}[(\q^\T \w)^2]}= 
    \frac{P_\theta |\q^\T \g|^2}
    {\q^\T \Cov_w \q},
    \label{eq:snr_q}
\end{equation}
where the numerator $\mathbb{E}[(\mathbf q^\T \mathbf g \theta)^2]$ is the post-processed signal power and the denominator $\mathbb{E}[(\mathbf q^\T \mathbf w)^2]$ is the corresponding effective noise power.

The LMMSE estimator in~\eqref{eq:active-q-estimate}, whose direction satisfies
$\q \propto \Cov_w^{-1}\g$ and maximizes this received SNR, yielding
\begin{equation}
    \SNR(\g,\F,\A)
    =
    P_\theta \g^\T \Cov_w^{-1}\g.
    \label{eq:snr_gFA}
\end{equation}

Without loss of generality, we normalize $P_\theta=1$ throughout the paper. Since \(\MMSE=(1+\SNR)^{-1}\)~\cite{li2006distribution}, minimizing MMSE is equivalent to maximizing the received SNR.

\begin{mdframed}
\begin{theorem}
(\textit{SNR maximization under average power constraints}) For fixed $(\g,\F,\A)$, the received SNR is
\[
\SNR(\g,\F,\A)
=
P_\theta \g^\T\Cov_w^{-1}(\F,\A)\g.
\]
Since $P_\theta$ is fixed, minimizing the MMSE is equivalent to
maximizing the $\g^\T\Cov_w^{-1}(\F,\A)\g$.
Accordingly, the joint design of $(\g,\F,\A)$ under the average
power constraints is formulated as
\begin{equation}
\begin{aligned}
&\max_{\g,\F,\A}\quad
\g^\T\Cov_w^{-1}(\F,\A)\g \notag\\
&\text{s.t.}\quad
\|\g\|^2
+\sigma_n^2\|\F\A\|_F^2
+\sigma_z^2\|\F\|_F^2
\le T\Ptx,\\
&\qquad
\|\A\g\|^2+
\mathrm{tr}(\A\Cov_w\A^\T)
\le T\Prx .
\end{aligned}
\tag{P1}
\label{eq:P1-active-feedback}
\end{equation}
where  $\Cov_w (\mathbf{F}, \mathbf{A}) \triangleq \sigma_n^2\big(\mathbf I+\mathbf F\mathbf A\big)\big(\mathbf I+\mathbf F\mathbf A\big)^{\mathsf T} + \sigma_z^2\mathbf F\mathbf F^{\mathsf T}$ with $\F$ strictly lower triangular and 
$\A$ lower triangular.
\label{thm:P1-active-feedback}
\end{theorem}
\end{mdframed}

\begin{proof}
For fixed $(\g,\F,\A)$, the LMMSE estimator $\q$ yields the MMSE--SNR relation \(\MMSE=(1+\SNR)^{-1}\)~\cite{li2006distribution}, so minimizing the MMSE is equivalent to maximizing $\g^\T\Cov_w^{-1}\g$. Expanding the average powers of the forward- and feedback-link channel inputs $\x$ and $\v$ yields the two constraints in \eqref{eq:P1-active-feedback}. For details, see Appendix~\ref{app:MMSE-SNR}. 
\end{proof}


The active linear feedback problem~\eqref{eq:P1-active-feedback} in Theorem~\ref{thm:P1-active-feedback} is challenging because the design variables $(\g, \F,\A)$ are nonlinearly coupled~\cite{butman1969general}. Although the precoder $\g$ can be optimized via a \emph{quadratically constrained quadratic program} (QCQP) when $(\F, \A)$ are fixed, the feedback coding matrices $\F$ and $\A$ appear both in the power constraints through the bilinear product $\F\A$ and in the objective through $\Cov_w^{-1}(\mathbf F,\mathbf A)$. As a result, the overall problem is nonconvex and does not admit a trivial closed-form solution. 

In the next section, we leverage this structure to simplify the problem. By first optimizing $\g$ for fixed $(\F,\A)$, we reduce the problem to a QCQP with a Rayleigh-quotient objective. This reformulation yields a KKT characterization and naturally leads to an envelope-gradient ascent approach for updating $(\F,\A)$ toward locally optimal active linear feedback schemes. 

\section{KKT-based Optimization and Eigenvector Characterizations} \label{sec:KKT-Optimization}
We now discuss an optimization algorithm for the active linear feedback problem~\eqref{eq:P1-active-feedback}. The key observation is that for fixed $(\mathbf F, \mathbf A)$, the covariance $\boldsymbol{\Sigma}^{-1}_w$ is fixed and the problem becomes a QCQP in the precoder $\mathbf g$. In contrast, optimizing over $(\mathbf F, \mathbf A)$ is nonconvex. This motivates a block-coordinate procedure: (i) update $\mathbf g$ optimally given $(\mathbf F, \mathbf A)$ using KKT conditions, and (ii) update $(\mathbf F, \mathbf A)$ using gradient ascent while enforcing causality and power feasibility.

\subsection{Updating $\mathbf g$ via KKT and Eigenvector Characterization}
With $(\mathbf F, \mathbf A)$ fixed, the $\mathbf g$-subproblem reduces to maximizing $\mathbf g^{\mathsf T}\boldsymbol{\Sigma}^{-1}_w\mathbf g$ subject to the transmitter and receiver average-power constraints, both of which are quadratic in $\mathbf g$. Let the residual power budgets be 
\begin{align}
\ctx(\mathbf F,\mathbf A)
&\triangleq T\Ptx-\sigma_n^2\|\mathbf F\mathbf A\|_F^2-\sigma_z^2\|\mathbf F\|_F^2,
\label{eq:residual forward-power budgets}
\\
\crx(\mathbf F,\mathbf A)
&\triangleq T\Prx-\mathrm{tr}\big(\mathbf A\boldsymbol{\Sigma}_w\mathbf A^{\mathsf T}\big),
\label{eq:residual feedback-power budgets}
\end{align}
so feasibility requires $\ctx>0$ and $\crx>0$. 

\begin{mdframed}
\textbf{Active Linear Feedback $\g$-subproblem (QCQP)} 
\begin{equation}\label{eq:P2-g-subproblem}
\begin{aligned}
\max_{\mathbf g}\quad & \mathbf g^{\mathsf T}\boldsymbol{\Sigma}_w^{-1}\mathbf g\\
\text{s.t.}\quad
& \|\mathbf g\|^2 \le \ctx,\quad
\|\mathbf A\mathbf g\|^2 \le \crx,\\
& \ctx> 0, \quad \crx > 0.
\end{aligned}
\tag{P2}
\end{equation}
\end{mdframed}


Introduce Lagrange multipliers $\lambda_1,\lambda_2 \geq 0$ for the two quadratic constraints, and define Lagrangian function 
\begin{equation} \label{eq:Lagrangian-g}
\begin{aligned}
\mathcal{L}(\mathbf g;\lambda_1,\lambda_2) &\triangleq
\mathbf g^{\mathsf T}\boldsymbol{\Sigma}_w^{-1}\mathbf g
\!-\! \lambda_1\!\Big(\|\mathbf g\|^2 \!-\! \ctx\Big)
\!-\! \lambda_2\!\Big(\|\mathbf A\mathbf g\|^2 \!-\! \crx\Big)\\
= ~&\underbrace{\g^{\T} (\Cov^{-1}_w -\lambda_1 \mathbf I - \lambda_2 \A^\T\A) \g}_{\text{the Lagrangian-adjusted quadratic form}} ~+~ \underbrace{\lambda_1 \ctx + \lambda_2 \
\crx}_{\text{constant}}
\end{aligned}
\end{equation}
At a Karush-Kuhn-Tucker (KKT) point, the multipliers satisfy complementary slackness,
\begin{equation} \label{eq:KKT-complementary slackness}
\lambda_1^\star(\|\mathbf g\|^2 - \ctx)=0, \quad 
\lambda_2^\star(\|\mathbf A\mathbf g\|^2 - \crx)=0,
\end{equation}
and the stationary condition $\nabla_{\g} \mathcal{L}(\g;\lambda_1,\lambda_2)=\mathbf{0}$ yields
\begin{equation} \label{eq:KKT-stationary-g}
\big(\Cov_w^{-1}-\lambda_2 \A^\T \!\A\big)\g= \lambda_1\g,
\end{equation}
which implies a generalized-eigenvector characterization of the optimal precoder $\g^\star$ direction. Equivalently, the Rayleigh-quotient structure of the QCQP yields an eigenvector characterization of $\g^\star$, formalized in the following theorem.



\begin{mdframed}
\begin{theorem}[Eigenvector characterization of optimal $(\g,\q)$]
\label{thm:active-gq-eig}
Fix any feasible $(\F,\A)$ and let $\Cov_w\triangleq \Cov_w(\F,\A)\succ 0$.
Let $\g^\star$ be a global maximizer of the $\g$-subproblem~\eqref{eq:P2-g-subproblem} and let $\lambda_1^\star,\lambda_2^\star$ denote the optimal Lagrange multipliers associated with the forward-/feedback-link power constraints.

1) The optimal precoder $\g^\star$ satisfies
\begin{equation}\label{eq:active-g-eigpair}
\begin{aligned}
\g^\star & \propto \mathbf v_{\max}\!\left(\Cov_w^{-1}-\lambda_2^\star \A^{\T}\A\right),\\
\lambda_1^\star &= \lambda_{\max}\!\left(\Cov_w^{-1}-\lambda_2^\star \A^{\T}\A\right),
\end{aligned}
\end{equation}
where $\mathbf v_{\max}(\cdot)$ denotes the dominant (largest-eigenvalue) eigenvector 
and $\lambda_{\max}(\cdot)$ denotes the largest eigenvalue.

2) The optimal LMMSE estimator $\q^\star$ satisfies
\begin{equation}\label{eq:active-q-propto}
\q^\star  \propto \Cov_w^{-1}\g^\star.
\end{equation}

\end{theorem}
\end{mdframed}

\begin{proof}
The KKT stationarity condition in \eqref{eq:KKT-stationary-g} and the Rayleigh-quotient characterization yield the dominant eigenvector form of $\g^\star$, while the LMMSE solution directly gives $\q^\star \propto \Cov_w^{-1}\g^\star$. For details, see Appendix~\ref{app:active-gq-eig}. 
\end{proof}


With $(\F, \A)$ fixed, computing the dominant eigenvector $\v_{\max}(\Cov_w^{-1}-\lambda_2^\star \A^{\T}\A)$ reduces to numerically determining the optimal Lagrange multiplier $\lambda_2^\star$. By rewriting the complementary slackness condition as a scalar equation in $\lambda_2$, this step reduces to a one-dimensional root-finding problem that can be solved efficiently via bisection. The precise result and its proof are given in Appendix~\ref{app:bisection-λ2-ratio}. 
Having characterized the global $\g$-update for fixed $(\F,\A)$, we next introduce an envelope-gradient ascent method for updating $(\F,\A)$.


\begin{algorithm}[hbt!]
\caption{KKTCode: Envelope Gradient Ascent}
\label{alg:optimal-active-feedback-scheme}

\KwIn{$T,~ \sigma_n^2,~ \sigma_z^2,~ \Ptx,~ \Prx$, $\{\eta_k\}$}
\KwOut{KKT stationary $(\g^\star,\F^\star,\A^\star)$, LMMSE $\q^\star$} 

\BlankLine
\textbf{Initialize:} choose a feasible causal $(\F^{(0)},\A^{(0)})$. \\

\BlankLine
\For{$k=0,1,\dots,K-1$}{
\textbf{(Inner)} Given $(\F^{(k)},\A^{(k)})$, find $\lambda_2^{(k)}$ by bisection:
\[r(\lambda_2^{(k)})=\dfrac{\crx(\F^{(k)},\A^{(k)})}{\ctx(\F^{(k)},\A^{(k)})}.\]
\[\text{Define}~~\mathbf M^{(k)} = \Cov_w^{-1}(\F^{(k)},\A^{(k)})-\lambda_2^{(k)} {\A^{(k)}}^{\mathsf T}\A^{(k)}.\]
\[\g^{(k)} \leftarrow\sqrt{ \ctx (\F^{(k)},\A^{(k)})/P_\theta}\,\mathbf v_{\max}(\mathbf M^{(k)})~~\text{and}\]
\[\lambda_1^{(k)} \leftarrow\lambda_{\max}(\mathbf M^{(k)}).\] 
\\
\BlankLine
\textbf{(Outer)} Compute the envelope gradients:
$\nabla_{\F}\mathcal V(\F^{(k)},\A^{(k)})$ and $\nabla_{\A}\mathcal V(\F^{(k)},\A^{(k)})$: 
\[
(\Delta\F^{(k)}, \Delta\A^{(k)}) \leftarrow \Pi_{\mathcal{S}}\big( (\nabla_{\F}\mathcal V,\;\nabla_{\A}\mathcal V)(\F^{(k)},\A^{(k)}) \big).
\]
Update $(\F,\A)$ via gradient ascent:
\[
(\F^{(k+1)},\A^{(k+1)}) \leftarrow (\F^{(k)},\A^{(k)}) + \eta_k(\Delta\F^{(k)},\Delta\A^{(k)}).
\]
}
\BlankLine
\textbf{Finalize:} Set $(\F^\star,\A^\star)\triangleq(\F^{(K)},\A^{(K)})$. Then solve the inner $\g$-subproblem at $(\F^\star,\A^\star)$ to obtain $(\g^\star,\lambda_1^\star,\lambda_2^\star)$, and $\q^\star=\q_{\rm LMMSE}(\g^\star,\F^\star,\A^\star)$ by~\eqref{eq:active-q-estimate}.

\BlankLine
\Return $(\g^\star,\F^\star,\A^\star,\q^\star)$.
\end{algorithm}

\subsection{Updating $(\mathbf F, \mathbf A)$ via Envelope Gradient Ascent}
The optimization of $(\F,\A)$ is formulated as maximizing the value function obtained after solving the $\g$-subproblem~\eqref{eq:P2-g-subproblem}:
\begin{equation}
\mathcal{V}(\F,\A) \triangleq \mathcal{L}(\g^\star(\F,\A);~\lambda_1^\star(\F,\A), \lambda_2^\star(\F,\A)).    
\end{equation}

For each feasible $(\F,\A)$, the $\g$-subproblem~\eqref{eq:P2-g-subproblem} yields the optimal primal--dual KKT solution 
$(\g^\star(\F,\A), ~\lambda_1^\star(\F,\A),  ~\lambda_2^\star(\F,\A)).$ Since $\mathcal{V}(\mathbf F, \mathbf A)$ represents the optimal value of a parametric optimization problem, the Envelope Theorem (Danskin's Theorem)~\cite{bertsekas1999nonlinear} allows us to compute its gradient with respect to $(\mathbf F, \mathbf A)$ using only the partial derivatives of the Lagrangian, without evaluating the sensitivity of the optimizer $\mathbf g^\star$. The closed-form expressions for $\nabla_F \mathcal{V}$ and $\nabla_A \mathcal{V}$ are given in Appendix~\ref{app:update-FA-envelop}
, and Algorithm~\ref{alg:optimal-active-feedback-scheme} uses them through the projected update direction $(\Delta \F,\Delta \A)$ with appropriately chosen stepsizes $\{\eta_k\}$. This procedure is summarized in Algorithm~\ref{alg:optimal-active-feedback-scheme}.

\subsection{Stationary Convergence of the Proposed Algorithm}

 We now analyze convergence via the reduced value function $\mathcal V(\F,\A)$ and projected first-order stationarity over the causal set. The following theorem shows that 
 the outer loop of Algorithm~\ref{alg:optimal-active-feedback-scheme} yields a monotone objective convergence and asymptotically vanishing projected gradient, implying KKT stationarity (up to sign symmetry). 

\begin{mdframed}
\begin{lemma} [Stationary Convergence of Algorithm~\ref{alg:optimal-active-feedback-scheme}] \label{cor:convergence}
Let $\Pi_{\mathcal S}$ denote the orthogonal projection onto the causal subspace of $(\F,\A)$. Assume that Algorithm~\ref{alg:optimal-active-feedback-scheme} maintains feasibility and that
$\mathcal V(\F,\A)$ is $L$-smooth. 

If the stepsizes satisfy $0<\eta_{\min}\le \eta_k \le 1/L$, then $\{\mathcal V(\F^{(k)},\A^{(k)})\}$ is nondecreasing and convergent, and
\[
\Big\|\Pi_{\mathcal S}\big(\nabla \mathcal V(\F^{(k)},\A^{(k)})\big)\Big\|_F \to 0.
\]
Consequently, every accumulation point $(\F^\star,\A^\star)$, together
with the optimal inner solution $\g^\star=\g^\star(\F^\star,\A^\star)$,
is a KKT-stationary solution of~\eqref{eq:P1-active-feedback}.
\end{lemma}
\end{mdframed}

\begin{proof}
By $L$-smoothness of $\mathcal V$ and orthogonality of $\Pi_{\mathcal S}$, each projected-gradient update yields a nonnegative increase in $\mathcal V$. Summing the resulting improvement bound gives $\|\Pi_{\mathcal S}\big(\nabla \mathcal V\big)\|_F \to 0$, so every accumulation point is KKT-stationary. For details, see Appendix~\ref{app:convergence}.
\end{proof}

Note that Lemma~\ref{cor:convergence} guarantees convergence of Algorithm~\ref{alg:optimal-active-feedback-scheme} to a KKT-stationary point whenever the stepsizes $\{\eta_k\}$ are bounded away from zero and not exceeding $1/L$.
With these convergence guarantees in hand, we now shift our focus to fundamental limits by deriving classical converse bounds and proving the asymptotic optimality of our proposed scheme.

\subsection{Optimal Passive Feedback from the Active Formulation}
To connect the general active formulation with the classical noisy-output-feedback setting, we consider the passive restriction, in which the receiver does not perform feedback coding. This corresponds to setting $\A=\I$, so the active problem~\eqref{eq:P1-active-feedback} reduces to the classical passive problem studied in~\cite{schalkwijk1966coding1, chance2011concatenated, mishra2023linear}. This restriction both simplifies the main problem and provides a globally optimal benchmark within the passive linear-feedback class.


\begin{mdframed}
\begin{theorem}
[Globally Optimal Passive Specialization of
Algorithm~\ref{alg:optimal-active-feedback-scheme}]
\label{thm:passive_GT_structure}
Under the passive restriction $\A=\I$ and $\lambda_2^\star=0$,
the passive specialization of
Algorithm~\ref{alg:optimal-active-feedback-scheme}
yields the globally optimal causal passive linear feedback scheme
for every finite blocklength $T$.
Moreover, the resulting scheme admits a
\textbf{geometric-Toeplitz} feedback structure, in which the
successive subdiagonal coefficients of the transmitter feedback
matrix form a geometric progression~
\cite{chance2011concatenated,agrawal2011iteratively}.

The optimal precoder and transmitter feedback matrix take the form
\begin{equation}
\g^\star
=
g_0
\begin{bmatrix}
1 & \beta^\star & (\beta^\star)^2 & \cdots & (\beta^\star)^{T-1}
\end{bmatrix}^\T,
\end{equation}
and
\begin{equation}
\F^\star
=
F_0
\begin{bmatrix}
0 & 0 & 0 & \cdots & 0\\
\beta^\star & 0 & 0 & \cdots & 0\\
(\beta^\star)^2 & \beta^\star & 0 & \cdots & 0\\
\vdots & \vdots & \ddots & \ddots & \vdots\\
(\beta^\star)^{T-1} &
(\beta^\star)^{T-2} &
\cdots &
\beta^\star & 0
\end{bmatrix},
\label{eq:GT_matrix_structure}
\end{equation}
where
\begin{equation}
F_0
=
-\frac{\sigma_n^2}{\sigma_n^2+\sigma_z^2}
\frac{1-(\beta^\star)^2}{(\beta^\star)^2},
\end{equation}
and $g_0>0$ is chosen such that the forward-link power
constraint is satisfied with equality.
The negative ratio $-\beta^\star$ also yields an SNR/MMSE-equivalent solution.

The ratio $\beta^\star\in(0,1)$ is the unique positive root of
\begin{equation}
\hspace{-0.7em}
h(\beta)\triangleq \Big[
\sigma_z^2
+T\Ptx\left(1+\frac{\sigma_z^2}{\sigma_n^2}\right)
+T\sigma_n^2
\Big]\beta^{2T}
-\sigma_n^2T\beta^{2T-2}
-\sigma_z^2.
\label{eq:h(beta)=0}
\end{equation}
Hence, the optimal passive scheme can be computed by
one-dimensional root finding with $\mathcal{O}(\log T)$ bisection iterations to polynomial accuracy.
\end{theorem}
\end{mdframed}

\begin{proof}
The projected KKT conditions under the passive restriction $\A=\I$ yield the geometric-Toeplitz structure, and the resulting scalar SNR optimization uniquely determines $\beta^\star$ up to sign through the condition $h(\beta)=0$. Each evaluation of $h(\beta)$ costs $O(\log T)$ via repeated squaring~\cite{knuth1998seminumerical}. Thus, for fixed accuracy, bisection finds $\beta^\star$ with $O(\log T)$ arithmetic complexity. For details, see Appendices~\ref{app:passive_GT_proof} and \ref{app:unique-toeplitz-optimal-ratio}. 
\end{proof}

Related geometric-Toeplitz structures were conjectured in Section IV-A in CL scheme~\cite{chance2011concatenated} and derived via an AM--GM argument in Section III-A of Agrawal~\emph{et al.}~\cite{agrawal2011iteratively}.
Here, the same structure is obtained from the KKT conditions of the passive specialization, and Theorem~\ref{thm:passive_GT_structure} shows that the resulting KKT-passive scheme is globally optimal for every finite \(T\). Although it retains a CL-style one-shot polynomial characterization, our derivation is fully non-asymptotic and therefore yields a different polynomial for the optimal passive structure. 


\section{SNR/MSE Converse Bounds and Achievability}
\label{sec:Achievability}


We now relate the proposed KKT-based design to fundamental SNR/MSE limits for causal linear feedback communication. In particular, for any causal linear feedback scheme, it is natural to ask what received SNR level is ultimately \emph{achievable} under the given forward- and feedback-link power budgets. For causal linear feedback schemes, the classical converse bounds are 

\vspace{-3.5mm}
\begin{align}
\hspace{-1em}
\SNR_{\rm active}  &\le \SNR_{\rm fw}+\SNR_{\rm fb}~{\small\text{(Elias--Butman)}}, \label{eq:SNR_active_UB}\\
\hspace{-1em}
\SNR_{\rm passive} &\le \SNR_{\rm fw}
+\frac{P_{\rm fw}}{\sigma_n^2+P_{\rm fw}}\SNR_{\rm fb}~{\small\text{(Chance--Love)}},   \label{eq:SNR_passive_UB}
\end{align}
where $\SNR_{\rm fw}\triangleq TP_{\rm fw}/\sigma_n^2$ and
$\SNR_{\rm fb}\triangleq TP_{\rm fb}/\sigma_z^2$.
The active converse is due to Elias and Butman~\cite{elias1967networks,butman1969general}.
The passive converse was given in error-exponent form by Kim \textit{et al.}~\cite{kim2011error}, and in received-SNR form for linear feedback by Chance and
Love~\cite{chance2011concatenated}. Derivations under our notation are provided in Appendix~\ref{app:converse bounds}.

By the MMSE-SNR identity \(\MMSE=(1+\SNR)^{-1}\)~\cite{li2006distribution}, these bounds imply the corresponding MMSE lower bounds. We next show that both converse bounds are asymptotically tight by proving that Algorithm~\ref{alg:optimal-active-feedback-scheme} achieves them as $T \to \infty$.   


\begin{mdframed} 
\begin{theorem} [Asymptotic Optimality of Algorithm~\ref{alg:optimal-active-feedback-scheme}] \label{thm:achievability} 
Let $\SNR_{\mathrm{KKT}}$ denote the received SNR achieved by Algorithm~\ref{alg:optimal-active-feedback-scheme}. Then it asymptotically achieves the corresponding converse bound:

(i) For active feedback,
\[
\lim_{T \to \infty} \frac{\SNR_{\mathrm{KKT}}}{\SNR_{\fw} + \SNR_{\fb}} = 1.
\]
(ii) Under the passive restriction $\A=\I$,
\[
\lim_{T \to \infty} \frac{\SNR_{\mathrm{KKT}}}{\SNR_{\fw} + \Big(\frac{\Ptx}{\sigma_n^2+\Ptx}\Big)\SNR_{\fb}}=1.
\]
\end{theorem}   
\end{mdframed}

\begin{proof} 
The achievability proof uses the asymptotically optimal one-time feedback construction of~\cite{tung2026achieving} as a feasible initialization. The proof combines a constant-gap feasible initialization with the monotone SNR improvement of Algorithm~\ref{alg:optimal-active-feedback-scheme}. The feasible construction yields an $O(1)$ additive gap to the Elias--Butman
converse that is uniformly bounded in $T$, and the algorithm
cannot increase this gap.
Since the converse grows linearly with $T$, the normalized gap
is $O(1/T)$ and hence vanishes, establishing first-order asymptotic optimality. For details, see Appendix~\ref{app:achievability}.
\end{proof}

Thus, the proposed KKT-based construction asymptotically attains both the active Elias--Butman converse and the passive Chance--Love converse. In the next section, we numerically validate the resulting optimality properties.  

\section{Numerical Results} \label{sec:Numerical results}
In this section, we present numerical simulations that validate the proposed KKT-based active and passive feedback schemes and illustrate their performance relative to the SK scheme~\cite{schalkwijk1966coding1}, the CL scheme~\cite{chance2011concatenated}, and the average-total-power variant of the DP scheme (DP-Avg)~\cite{mishra2023linear}. We consider a wide range of feedback-noise variance ratios $\sigma_z^2/\sigma_n^2$, and different blocklengths $T$, under fixed average power constraints $\Ptx$ and $\Prx=\Ptx+\sigma_n^2$, ensuring a fair comparison between active and passive feedback schemes under the same feedback-link power budget. 


\subsection{SNR Comparison over Feedback–Noise}

We first reproduce the simulation setup in \cite[Fig.~10]{chance2011concatenated} to enable a fair comparison across all considered linear feedback schemes. We fix the per-symbol average transmit-power constraint to $\Ptx = 1$ and set $\sigma_n^2 = 1$, while sweeping the
feedback–noise variance ratio $\sigma_z^2/\sigma_n^2$. 

The received SNR is shown in Fig.~\ref{fig:SNR Comparison-rho=1}. Across all feedback–noise levels and both blocklengths $T \in \{5,10\}$, the proposed active-feedback scheme achieves the highest SNR among all compared schemes. By the MMSE--SNR identity~\cite{li2006distribution}, this also corresponds to the lowest MSE, showing a clear finite-blocklength gain over passive feedback designs. Among the noisy passive-feedback schemes, the proposed optimal passive scheme (KKT-Passive) essentially matches the best feasible SNR and remains nearly identical to the DP-Avg scheme, but with much lower design complexity.  


\begin{figure} [hbt!]
    \vspace{-3.0 mm}
    \centering
    \includegraphics[width=0.9\linewidth]{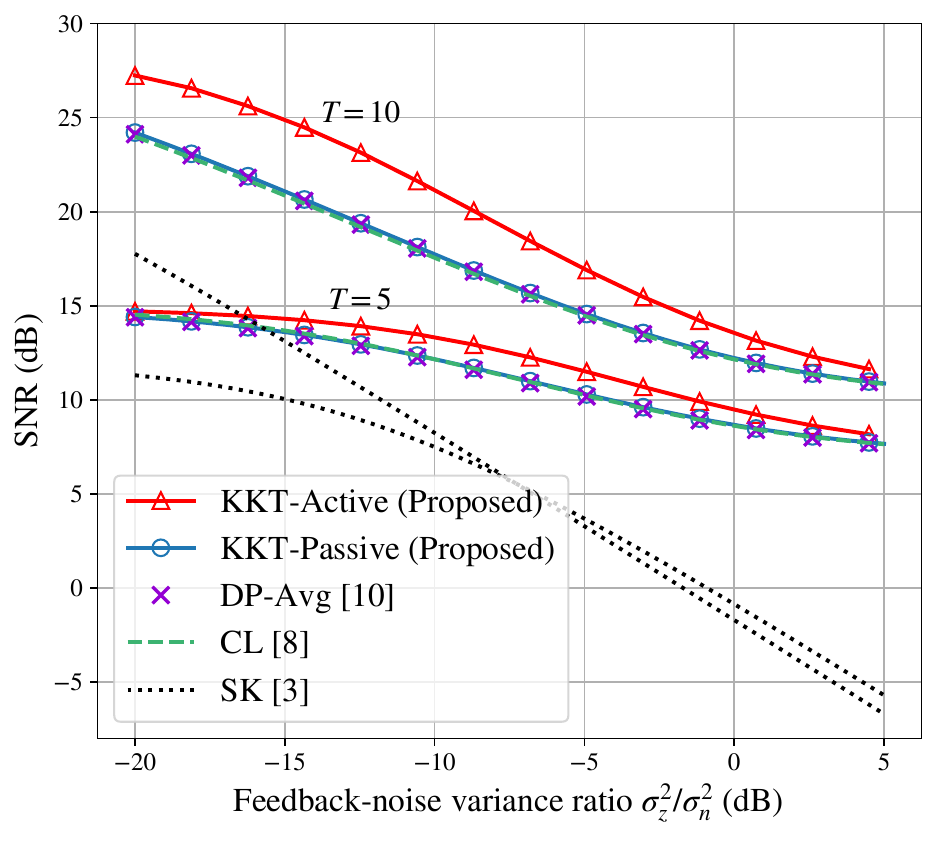}
    \vspace{-3.0 mm}
    \caption{SNR comparison versus feedback-noise variance ratio for \(T=5,10\).}
    \label{fig:SNR Comparison-rho=1}
\end{figure}

\begin{figure}[hbt!]
    \vspace{-3.0 mm}
    \centering
    \includegraphics[width=0.9\linewidth]{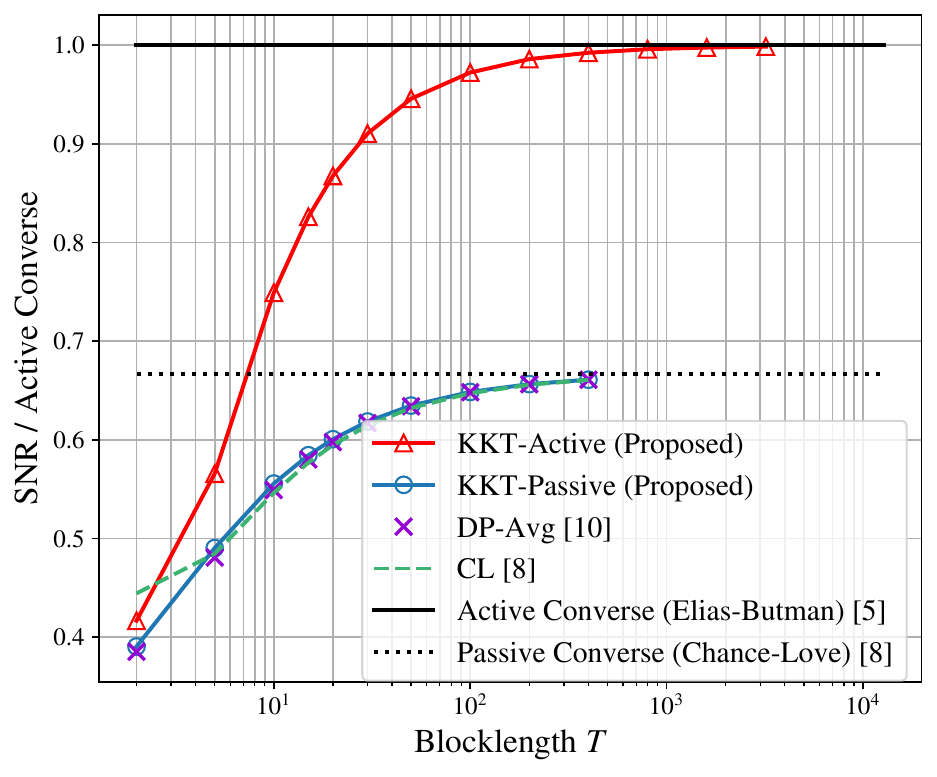}
    \vspace{-3.0 mm}
    \caption{Achievability of the SNR converse bound versus $T$ $(\sigma_z=\sigma_n=1.0)$.}
    \label{fig:Achievability-SNR}
\end{figure}

\subsection{Achievability of the Elias--Butman Converse}

For the achievability study, we set symmetric forward and feedback noise levels to $\sigma_n=\sigma_z=1.0$ and evaluate the schemes as the blocklength $T$ increases.

In Fig.~\ref{fig:Achievability-SNR}, we compare the proposed active and passive schemes with CL~\cite{chance2011concatenated}, and DP-Avg~\cite{mishra2023linear}. To directly assess how closely each design approaches the fundamental linear-feedback limit, the received SNR of every scheme is normalized by the Elias--Butman active converse, while the Chance--Love passive converse is also shown as a reference for the passive class. The proposed active scheme approaches the Elias--Butman converse bound quickly and achieves the best SNR/MSE performance across the tested blocklengths. In contrast, all passive schemes, including the proposed optimal passive scheme, CL, and DP-Avg, are limited by the passive Chance--Love converse. This confirms that the KKT-based design is essentially optimal within its class, while active feedback provides a strict additional SNR gain.

\section{Conclusion} \label{sec:Conclusion}
In this paper, we studied causal linear feedback schemes for AWGN channels with noisy output feedback and developed a constructive KKT-optimal active feedback design that asymptotically achieves the Elias--Butman SNR converse bound, thereby establishing MSE/SNR optimality over the general causal linear class. We further showed that the passive noisy output-feedback problem arises as a special case of the active formulation, and that its globally optimal solution admits a geometric-Toeplitz structure characterized by a one-shot polynomial root-finding procedure with \(O(\log T)\) complexity.

\bibliographystyle{IEEEtran}
\bibliography{Refs/Linear_Feedback, Refs/Nonlinear_Feedback, Refs/Wireless_Comm, Refs/CS_Algorithms}

\appendices

\section{MMSE--SNR Relationship}
\label{app:MMSE-SNR}

\paragraph{MMSE-SNR Relationship}
A causal linear active-feedback scheme is written as:
\begin{align}
&\mathbf x=\mathbf g\,\theta+ \F\A \mathbf{n} + \F \mathbf{z}, \\
&\mathbf v=\mathbf A\,\mathbf y,
\label{eq:active-linear-law}
\end{align}
where $\mathbf g\in\mathbb R^{T}$ is the precoding vector, $\mathbf F\in\mathbb R^{T\times T}$  is the transmitter feedback coding matrix (transmitter filter), and $\mathbf A\in\mathbb R^{T\times T}$ is the receiver feedback coding matrix (receiver filter). Causality requires $\F$ to be strictly lower triangular and $\A$ to be lower triangular. Equivalently, the received vector can be written as:
\begin{equation}
    \mathbf{y} = \g \theta + \w, \quad \w \triangleq (\I+\F\A)\n+\F\z,
    \label{eq:active-block-y}
\end{equation}
where $\w$ denotes the effective noise vector with covariance 
\begin{equation}
\Cov_w\triangleq \sigma_n^2\big(\mathbf I+\mathbf F\mathbf A\big)\big(\mathbf I+\mathbf F\mathbf A\big)^{\mathsf T} + \sigma_z^2\mathbf F\mathbf F^{\mathsf T}.
\end{equation}

We also associate to $\mathbf{q}$ the received (post-processed) SNR
\begin{equation}
    \mathrm{SNR}(\g, \F, \A, \q)
    := \frac{\mathbb{E}[(\mathbf{q}^\T\mathbf{g}\,\theta)^2]}
           {\mathbb{E}[(\mathbf{q}^\T\mathbf{w})^2]}
    = \frac{\mathbb{E}[\theta^2]|\mathbf{q}^\T\mathbf{g}|^2}
           {\mathbf{q}^\T \boldsymbol{\Sigma}_w \mathbf{q}},  \label{eq:SNR-q}
\end{equation}
where, by the effective signal-noise relation in~\eqref{eq:active-block-y}, the numerator $\mathbb{E}[(\mathbf q^\T \mathbf g \theta)^2]$ is the post-processed signal power and the denominator $\mathbb{E}[(\mathbf q^\T \mathbf w)^2]$ is the corresponding effective noise power.

For any linear estimator $\mathbf{q}$ we consider the linear estimator
$\hat{\theta} = \mathbf{q}^\T \mathbf{y}$ and define the mean–squared error as
\begin{align}
    \mathrm{MSE}(\g, \F, \A, \q)
    &:= \mathbb{E}\bigl[(\theta - \hat{\theta})^2\bigr]
      = \mathbb{E}\bigl[(\theta - \mathbf{q}^\T \mathbf{y})^2\bigr] \notag \\
    &= \mathbb{E}\!\left[\bigl((1-\mathbf{q}^\T \mathbf{g})\theta
         - \mathbf{q}^\T \mathbf{w}\bigr)^2\right] \notag \\
    &= (1-\mathbf{q}^\T \mathbf{g})^2 \,P_\theta
       + \mathbf{q}^\T \boldsymbol{\Sigma}_w \mathbf{q}.    \label{eq:MSE-q}
\end{align}

Since $\y=\g\theta + \w$ is jointly Gaussian with $\theta$, the MSE-optimal estimator is the LMMSE estimator. Then
\begin{equation}
    \begin{aligned}
    \mathbf{q}_{\mathrm{LMMSE}} &= (\boldsymbol{\Sigma}_{\theta \mathbf{y}} \boldsymbol{\Sigma}_{\mathbf{y}\mathbf{y}}^{-1})^\T = \boldsymbol{\Sigma}_{\mathbf{y}\mathbf{y}}^{-1}P_\theta\mathbf{g}
    = \frac{\boldsymbol{\Sigma}_w^{-1}\mathbf{g}}
           {1/P_\theta+\mathbf{g}^\T \boldsymbol{\Sigma}_w^{-1}\mathbf{g}}.
    \end{aligned}
    \label{eq:q_LMMSE}
\end{equation}
Note that \eqref{eq:q_LMMSE} also follows by differentiating \eqref{eq:SNR-q} with respect to $\q$ and setting the gradients to zero.

Substituting \eqref{eq:q_LMMSE} into \eqref{eq:SNR-q} and \eqref{eq:MSE-q}, the corresponding MMSE and received SNR are
\begin{equation}
    \begin{aligned}
    \mathrm{MMSE}(\g, \F, \A)&= \min_{\mathbf{q}}\mathrm{MSE}(\mathbf{q}) = \mathrm{MSE}(\mathbf{q}_{\mathrm{LMMSE}})\\ 
    &= \frac{1}{1/P_\theta+\mathbf{g}^\T \boldsymbol{\Sigma}_w^{-1}\mathbf{g}},
    \end{aligned}
    \label{eq:MMSE-g}
\end{equation}
\begin{equation}
    \SNR(\g, \F, \A)=P_\theta \; \mathbf{g}^\T
    \boldsymbol{\Sigma}_w^{-1}\mathbf{g}, \qquad \qquad \qquad
    \label{eq:SNR-g}
\end{equation}
so that the SNR and the MMSE are related by
\begin{equation}
    \mathrm{MMSE} = \frac{P_\theta}{1+\mathrm{SNR}}. \label{eq:MMSE-SNR}
\end{equation}

This MMSE-SNR relationship aligns with classical wireless communication theory~\cite{li2006distribution}. Moreover, because the received SNR depends only on the direction of the linear estimator, a simple rescaling of the (biased) MMSE estimator can always be used to enforce an unbiased one by an appropriate scalar rescaling~\cite{cioffi1995mmse}. The minimum-variance unbiased (MVU) and LMMSE estimators share the same direction and thus yield identical SNR. For ease of comparison with existing works, we therefore use the MSE, rather than the bias, as  the primary performance metric.

\paragraph{Causal Active Linear Feedback Design Problem (Equivalent SNR Maximization)} 
Since $P_\theta$ is fixed, minimizing $\MMSE(\g,\F,\A)$ is equivalent to maximizing
$\SNR(\g,\F,\A)$, and equivalently to maximizing the normalized objective $\g^\T \Cov_w^{-1}(\F,\A) \g$ used in Theorem~\ref{thm:P1-active-feedback}.

Finally, the average forward- and feedback-link power constraints can be
rewritten explicitly in terms of $(\g, \F, \A)$. From \eqref{eq:active-power-constraints},
\begin{align}
\mathbb{E}\!\left[\|\x\|^2\right]
&=
\mathbb{E}\!\left[\| \g \theta+ \F \A \n+ \F \z\|^2\right] \notag\\
&=
P_\theta\|\g\|^2 + \sigma_n^2\| \F\A\|_F^2 + \sigma_z^2\|\F\|_F^2
\le T P_{\mathrm{fw}}.
\label{eq:active-transmitter-power-constraint}
\end{align}
Also, since $\v=\A\y$ and $\y=\g\theta+\w$,
\begin{align}
\mathbb{E}\!\left[\|\v\|^2\right]
&=
\mathbb{E}\!\left[\|\A\y\|^2\right] \notag\\
&=
P_\theta\|\A\g\|^2 + \mathrm{tr}(\A\Cov_w \A^\T)
\le T P_{\mathrm{fb}}.
\label{eq:active-receiver-power-constraint}
\end{align}
Thus, after eliminating the decoder variable $q$ through the LMMSE solution,
the joint design problem reduces to
\[
\max_{\g, \F, \A}\; \g^\T \Cov_w^{-1}(\F,\A) \g
\]
subject to \eqref{eq:active-transmitter-power-constraint}, \eqref{eq:active-receiver-power-constraint}, with $\F$ strictly lower triangular and $\A$
lower triangular, which is exactly the formulation in Theorem~\ref{thm:P1-active-feedback}.
\hfill $\blacksquare$

\section{Proof of Theorem~\ref{thm:active-gq-eig}: Necessary optimality of $(\g,\q)$}
\label{app:active-gq-eig}
The KKT stationarity condition for the Lagrangian is equivalent to \eqref{eq:KKT-stationary-g} after rearranging terms. Specifically, it implies that $\g^\star$ is an eigenvector of the symmetric matrix $\Cov_w^{-1}-\lambda_2 \A^\T\A$ with eigenvalue $\lambda_1^\star$. 

Note that the objective $\g^{\mathsf T}\Cov_w^{-1}\g$ is increasing under scaling of $\g$.
Hence, whenever the problem is feasible, an optimum cannot lie in the interior of the transmit-power constraint; equivalently, complementary slackness~\eqref{eq:KKT-complementary slackness}
implies that if $\lambda_1^\star>0$ then
\[
\|\g^\star\|^2= \ctx.
\]
With the norm fixed, optimizing $\g$ reduces to choosing its \emph{direction}.
In particular, for fixed $\lambda_2^\star$, the $\g$-dependent part of the Lagrangian can be written as
\[
\mathcal{L}(\g;\lambda_1^\star,\lambda_2^\star)
= \g^{\mathsf T}\!\Big(\Cov_w^{-1}-\lambda_1^\star\mathbf I-\lambda_2^\star \A^{\mathsf T}\A\Big)\g+\text{const.}
\]
Therefore, maximizing the Lagrangian over the direction of $\g$ under the fixed-norm constraint is equivalent to maximizing a Rayleigh quotient of the
symmetric matrix $\Cov_w^{-1}-\lambda_1^\star \mathbf I-\lambda_2^\star \A^{\mathsf T}\A$.
By the Rayleigh-quotient characterization, any maximizer aligns with its dominant eigenvector.
Moreover, since shifting by $-\lambda_1^\star \mathbf I$ does not change eigenvectors, the maximizing direction is equivalently given by
$\mathbf v_{\max}\!\big(\Cov_w^{-1}-\lambda_2^\star \A^{\T}\A\big)$. Finally, the KKT stationarity condition $\nabla_{\g}\mathcal{L}(\g^\star;\lambda_1^\star,\lambda_2^\star)=\mathbf 0$ confirms that $\lambda_1^\star=\lambda_{\max}\!\left(\Cov_w^{-1}-\lambda_2^\star \A^{\T}\A\right)$ and establishes~\eqref{eq:active-g-eigpair}.

For the linear Gaussian model $\mathbf y=\g^\star\theta+\mathbf w$ with
$\mathbf w\sim\mathcal N(\mathbf 0,\Cov_w)$, the SNR-maximizing (equivalently MMSE-minimizing) linear estimator is the LMMSE combiner,
whose direction is proportional to $\Cov_w^{-1}\g^\star$ by~\eqref{eq:q_LMMSE}, establishing~\eqref{eq:active-q-propto}.

If the receiver constraint is inactive (i.e., $\mathbf A = \mathbf I$, passive feedback scheme), complementary slackness implies $\lambda_2^\star=0$, and the $\g$-subproblem reduces to
$\max_{\|\g\|^2\le \ctx}\ \g^{\T}\Cov_w^{-1}\g$.
By the Rayleigh-quotient characterization, the maximizer aligns with the dominant eigenvector of $\Cov_w^{-1}$. \qed

\section{Computing $\lambda_2^\star$ via Bisection Method}
\label{app:bisection-λ2-ratio}

With $(\F, \A)$ fixed, the only remaining step to obtain the dominant eigenvector $\v_{\max}(\Cov_w^{-1}-\lambda_2^\star \A^{\T}\A)$ is to numerically determine the optimal Lagrange multiplier $\lambda_2^\star$. The key idea is to convert the power constraints enforced by complementary slackness into a scalar equation in $\lambda_2$ involving the ratio of the residual power budgets.

\begin{mdframed}
\begin{theorem}[Computing $\lambda_2^\star$ via bisection method]
\label{thm:bisection-λ2-ratio}
Fix any feasible $(\F,\A)$ and let $\Cov_w\triangleq \Cov_w(\F,\A)\succ 0$. For each $\lambda_2\geq0$, let
\begin{equation} \label{eq:v_max(λ2)}
\v_{\max}(\lambda_2) \triangleq \v_{\max}\!\left(\Cov_w^{-1}-\lambda_2 \A^{\T}\A\right)
\end{equation}
denote a dominant eigenvector with unit norm.

Define the scalar ratio function
\begin{equation} \label{eq:ratio-r(λ2)}
r(\lambda_2)\ \triangleq\ \frac{\|\A \g(\lambda_2)\|^2}{\|\g(\lambda_2)\|^2}
\ =\\  \|\A \v_{\max}(\lambda_2)\|^2, 
\end{equation}
where $\g(\lambda_2)$ is any nonzero scaling of $\v_{\max}(\lambda_2)$.

If the receiver constraint is active at the optimum of~\eqref{eq:P2-g-subproblem} (equivalently, $\lambda_2^\star>0$),
then $\lambda_2^\star$ satisfies the scalar equation
\begin{equation}\label{eq:opt-ratio-r(λ2*)}
r(\lambda_2^\star)\ =\ \frac{\crx(\F,\A)}{\ctx(\F,\A)}.
\end{equation}
Since $r(\lambda_2)$ is nonincreasing in $\lambda_2$, the optimal $\lambda_2^\star$ can be found by bisection on~\eqref{eq:opt-ratio-r(λ2*)}.
\end{theorem}
\end{mdframed}

Since the power ratio $\|\A \g\|^2/\|\g\|^2$ is invariant to any nonzero scaling of $\g^\star$ and the eigenvector has unit norm, \eqref{eq:ratio-r(λ2)} simplifies to $\|\A\v_{\max}(\lambda_2)\|^2$.

At an optimum of~\eqref{eq:P2-g-subproblem}, complementary slackness implies that when $\lambda_2^\star>0$, the receiver constraint is tight
\[
\|\A\g^\star(\lambda_2^\star)\|^2=\frac{\crx(\F, \A)}{P_\theta}.
\]
For the same optimum, the objective $\g^{\mathsf T}\Cov_w^{-1}\g$ is increasing under positive scaling of $\g$, so the transmit constraint must also be tight
\[
\|\g^\star(\lambda_2^\star)\|^2=\frac{\ctx(\F, \A)}{P_\theta}.
\]
Taking the ratio yields 
\[
\frac{\|\A\g^\star(\lambda_2^\star)\|^2}{\|\g^\star(\lambda_2^\star)\|^2} =  \frac{\crx(\F,\A)}{\ctx(\F,\A)}.
\]
By Theorem~\ref{thm:active-gq-eig}, the optimal precoder satisfies $\g^\star \propto \v_{\max}(\lambda_2^\star)$, hence the left-hand side equals $r(\lambda_2^\star)$, which proves the target equation~\eqref{eq:opt-ratio-r(λ2*)}.

The monotonicity of $r(\lambda_2)$ follows directly from the variational characterization of the Lagrangian-adjusted objective. Specifically, let
\[
J(\v;\lambda_2) \triangleq \v^\T \big(\Cov_w^{-1} - \lambda_2 \A^{\mathsf T}\A\big)\v.
\]
By definition, $\v_{\max}(\lambda_2)$ maximizes $J(\v;\lambda_2)$ over the unit sphere $\|\v\|=1$. Consider arbitrary multipliers $\lambda_b > \lambda_a \geq 0$, and let $\v_a=\v_{\max}(\lambda_a)$ and $\v_b=\v_{\max}(\lambda_b)$. The optimality of each eigenvector implies 
\begin{align*}
J(\v_b; \lambda_b) &\ge J(\v_a; \lambda_b), \\
J(\v_a; \lambda_a) &\ge J(\v_b; \lambda_a).
\end{align*}
Summing the two inequalities cancels the $\Cov_w^{-1}$ terms and yields 
\[
(\lambda_b-\lambda_a) \|\A \v_a\|^2\geq (\lambda_b-\lambda_a)\|\A \v_b\|^2.
\]
Since $\lambda_b-\lambda_a >0$, we conclude that $\|\A \v_a\|^2 \geq \|\A \v_b\|^2$, or equivalently $r(\lambda_a)\geq r(\lambda_b)$ for all $\lambda_b \geq \lambda_a$. This monotonicity guarantees that the solution $\lambda_2^\star$ to the scalar equation $r(\lambda_2^\star)= \crx/ \ctx$ is unique, allowing it to be efficiently found by the bisection method. \qed

\section{Proof of Envelope Gradients}
\label{app:update-FA-envelop}

In this appendix, we fix an arbitrary feasible pair $(\F,\A)$ and write
\begin{equation}
\Cov_w \triangleq \Cov_w(\F,\A)
= \sigma_n^2(\I+\F\A)(\I+\F\A)^{\T} + \sigma_z^2 \F\F^{\T} \succ 0.
\label{eq:app:Sigmaw-def}
\end{equation}
Let $(\g^\star,\lambda_1^\star,\lambda_2^\star)$ be the optimal primal--dual solution of the
$\g$-subproblem (P2) for this fixed $(\F,\A)$ (with $P_\theta=1$), and define the value function
\(
\mathcal V(\F,\A)\triangleq \max_{\g}\ \g^{\T}\Cov_w^{-1}\g
\)
subject to the two quadratic constraints in (P2).

\begin{mdframed}[
  innerleftmargin=5pt,
  innerrightmargin=5pt,
]
\begin{theorem}[Computing $\nabla \mathcal{V}(\F,\A)$ via the Envelope Theorem~\cite{bertsekas1999nonlinear}]
\label{thm:update-FA-envelop}
Let $(\mathbf g^\star, \lambda_1^\star, \lambda_2^\star)$ be the optimal primal-dual solution to the $\mathbf g$-subproblem~\eqref{eq:P2-g-subproblem} for a given feasible pair $(\mathbf F, \mathbf A)$ by Theorems~\ref{thm:active-gq-eig}-\ref{thm:bisection-λ2-ratio}. The value function $\mathcal{V}(\mathbf F, \mathbf A)$ is differentiable, and its gradients are given by the partial derivatives of the Lagrangian $\mathcal{L}$ in \eqref{eq:Lagrangian-g} evaluated at the optimum:
\begin{equation} \label{eq:grad-V-envelope}
\nabla_{\mathbf X} \mathcal{V}(\mathbf F, \mathbf A) = \nabla_{\mathbf X} \mathcal{L}(\mathbf g^\star; \mathbf F, \mathbf A, \lambda_1^\star, \lambda_2^\star),~\text{for } \mathbf{X} \in \{\mathbf F, \mathbf A\}.
\end{equation}

Specifically, define the rank-one auxiliary matrix
\[
\mathbf S^\star \triangleq  \Cov_w^{-1}\mathbf g^\star ( \Cov_w^{-1}\mathbf g^\star)^\T =\Cov_w^{-1}\mathbf g^\star{\mathbf g^\star}^{\T}\Cov_w^{-1}.
\]
Then the gradients admit the following closed forms:
\begin{align}
\nabla_{\F} \mathcal{V} &= -2\sigma_n^2\mathbf S^\star(\I+\F\A)\A^\T - 2\sigma_z^2\mathbf S^\star\F \nonumber\\
&-\lambda_1^\star[2\sigma_n^2\F\A\A^T+2\sigma_z^2\F] \nonumber\\
&-\lambda_2^\star[2\sigma_n^2\A^\T\A(\I+\F\A)\A^\T + 2\sigma_z^2\A^\T\A\F], \label{eq:gradV-F}
\\
\nabla_{\A} \mathcal{V} &= -2\sigma_n^2\F^\T \mathbf S^\star   (\I+\F\A) \nonumber\\
&-\lambda_1^\star[2\sigma_n^2\F^\T \F\A] \nonumber\\
&-\lambda_2^\star[2\A\Cov_w+2\sigma_n^2\F^\T \A^\T\A(\I+\F\A) + 2\A \g^\star{\g^\star}^\T]. \label{eq:gradV-A}
\end{align}

Let $\mathcal S$ denote the causal subspace of feasible pairs $(\F,\A)$ (strictly lower-triangular $\F$ and lower-triangular $\A$), and $\Pi_{\mathcal{S}}$ be the orthogonal projection onto $\mathcal{S}$ under the Frobenius inner product. We take the update direction as 

\begin{equation}
    (\Delta \F, \Delta \A) = \Pi_{\mathcal{S}}(\nabla_{\mathbf F} \mathcal{V}, \; \nabla_{\mathbf A} \mathcal{V}),
\end{equation}
i.e., $\Delta \mathbf F = \scalebox{0.9}{$\mathrm{StrictlyLower}$}(\nabla_{\mathbf F} \mathcal{V})$ and $\Delta \mathbf A = \scalebox{0.9}{$\mathrm{Lower}$}(\nabla_{\mathbf A} \mathcal{V})$.

\end{theorem}
\end{mdframed}

Since (P2) is a maximization over a nonempty compact feasible set, at least one maximizer $\g^\star$ exists.
Moreover, for fixed $(\g,\lambda_1,\lambda_2)$ the Lagrangian in \eqref{eq:Lagrangian-g}
is continuously differentiable in $(\F,\A)$ through $\Cov_w(\F,\A)$, $\ctx(\F,\A)$, and
$\crx(\F,\A)$.
Therefore, Danskin's theorem (envelope theorem)~\cite{bertsekas1999nonlinear} implies that the value function
$\mathcal V(\F,\A)$ is directionally differentiable and its directional derivative depends only on the
partial derivatives of the Lagrangian evaluated at maximizers, i.e.,
the terms involving the sensitivity $d\g^\star$ do not appear.
Then the envelope gradients are given by
\[
\nabla_{\X}\mathcal V(\F,\A)=\nabla_{\X}\mathcal L(\g^\star;\F,\A,\lambda_1^\star,\lambda_2^\star),
\qquad \X\in\{\F,\A\},
\]
namely, we evaluate the partial derivatives of $\mathcal L$ at an optimal primal--dual solution
without differentiating the maximizer $\g^\star(\F,\A)$ with respect to $(\F,\A)$.

\subsection{Matrix Differential and Trace Identities}
For real matrices $\F$ and $\A$, we compute gradients using differentials and the Frobenius inner product:
for a scalar function $\phi(\F,\A)$,
\[
d\phi = \langle \nabla_{\F}\phi,\, d\F\rangle + \langle \nabla_{\A}\phi,\, d\A\rangle,
~\text{where}~
\langle \mathbf{U},\mathbf{V}\rangle \triangleq \mathrm{tr}(\mathbf{U}^{\T}\mathbf{V}).
\]
We repeatedly use the Leibniz product rule and the resulting identity for a matrix inverse (obtained by applying the product rule to $\I=\Cov_w \Cov_w^{-1}$ for any invertible $\Cov_w$):
\begin{align}
d(\mathbf{U}\mathbf{V}) &= (d\mathbf{U})\mathbf{V}+\mathbf{U}(d\mathbf{V}), \qquad \text{(product rule)} \label{eq:app:dUV}\\
d(\Cov_w^{-1}) &= -\Cov_w^{-1}(d\Cov_w)\Cov_w^{-1}, \qquad \text{(inverse rule)} \label{eq:app:dinv}
\end{align}
and the cyclic property of the trace:
\begin{equation}
\mathrm{tr}(\mathbf{U}\mathbf{V}\mathbf{W})=\mathrm{tr}(\mathbf{W}\mathbf{U}\mathbf{V})=\mathrm{tr}(\mathbf{V}\mathbf{W}\mathbf{U}),~~ \text{(cyclic property)}
\end{equation}
where $\mathbf{U},\mathbf{V},\mathbf{W}$ are any conformable real matrices.

\subsection{Objective term: gradients of $\g^{\star\T}\Cov_w^{-1}\g^\star$}
Define 
\[
\S^\star \triangleq \Cov_w^{-1}\g^\star{\g^\star}^{\T}\Cov_w^{-1}=\Cov_w^{-1}\g^\star(\Cov_w^{-1}\g^\star)^\T.
\]
Using \eqref{eq:app:dinv} and treating $\g^\star$ as constant by the envelope theorem,
\begin{align}
&d\big(\g^{\star\T}\Cov_w^{-1}\g^\star\big)
= \mathrm{tr}\!\left({\g^\star}^{\T} d(\Cov_w^{-1}) \g^\star\right)
= \mathrm{tr}\!\left(d(\Cov_w^{-1})\,\g^\star{\g^\star}^{\T}\right)\nonumber\\
&= -\mathrm{tr}\!\left(\Cov_w^{-1}(d\Cov_w)\Cov_w^{-1}\,\g^\star{\g^\star}^{\T}\right)
= -\mathrm{tr}\!\left(\S^\star\, d\Cov_w\right).
\label{eq:app:dobj}
\end{align}
Next, write $\M\triangleq \I+\F\A$. From \eqref{eq:app:Sigmaw-def},
\begin{equation}
\begin{aligned}
&d\Cov_w
= \sigma_n^2\big((d\M)\,\M^{\T} + \M\,(d\M)^{\T}\big)
 + \sigma_z^2\big((d\F)\,\F^{\T}+\F\,(d\F)^{\T}\big),\\
\end{aligned}
\label{eq:app:dSigma}
\end{equation}
where $d\M = (d\F)\A+\F (d\A)$.

\paragraph{Derivative of $\g^{\star\T}\Cov_w^{-1}\g^\star$ with respect to $\F$}
Keeping only the $d\F$ terms in \eqref{eq:app:dSigma} gives  $d\M|_{d\F}=d\F \A$
\[
d\Cov_w|_{d\F}
=\sigma_n^2\Big((d\F)\A\M^{\T}+\M\big(d\F \A\big)^{\T}\Big)
+\sigma_z^2\Big((d\F)\,\F^{\T}+\F\,(d\F)^{\T}\Big).
\]
Substituting into \eqref{eq:app:dobj} and using trace cyclicity yields
\[
d\big(\g^{\star\T}\Cov_w^{-1}\g^\star\big)|_{d\F}
= \left\langle
-2\sigma_n^2\,\S^\star \M \A^{\T} - 2\sigma_z^2\,\S^\star \F,\ d\F
\right\rangle.
\]
Hence,
\begin{equation}
\nabla_{\F}\big(\g^{\star\T}\Cov_w^{-1}\g^\star\big)
= -2\sigma_n^2\,\S^\star(\I+\F\A)\A^{\T} - 2\sigma_z^2\,\S^\star\F.
\label{eq:app:gradF-obj}
\end{equation}

\paragraph{Derivative of  $\g^{\star\T}\Cov_w^{-1}\g^\star$ with respect to $\A$}
Keeping only the $d\A$ terms in \eqref{eq:app:dSigma} gives $d\M|_{d\A}=\F (d\A)$ and
\[
d\Cov_w|_{d\A}=\sigma_n^2\Big(\F (d\A)\M^{\T} + \M(\F (d\A))^{\T}\Big).
\]
Substituting into \eqref{eq:app:dobj} yields
\[
d\big(\g^{\star\T}\Cov_w^{-1}\g^\star\big)|_{d\A}
= \left\langle
-2\sigma_n^2\,\F^{\T}\S^\star \M,\ d\A
\right\rangle,
\]
and thus
\begin{equation}
\nabla_{\A}\big(\g^{\star\T}\Cov_w^{-1}\g^\star\big)
= -2\sigma_n^2\,\F^{\T}\S^\star(\I+\F\A).
\label{eq:app:gradA-obj}
\end{equation}

\subsection{Power Budget terms: gradients of $\ctx(\F,\A)$ and $\crx(\F,\A)$}
Recall
\begin{equation*}
\begin{aligned}
&\ctx(\F,\A)=T\Ptx-\sigma_n^2\|\F\A\|_F^2-\sigma_z^2\|\F\|_F^2,\\
&\crx(\F,\A)=T\Prx-\mathrm{tr}(\A\Cov_w\A^{\T}).
\end{aligned}
\end{equation*}

\paragraph{Gradients of $\ctx(\F,\A)$}
Using $d\|\X\|_F^2=2\,\mathrm{tr}(\X^{\T}d\X)$ and the product rule~\eqref{eq:app:dUV},
\begin{align}
\nabla_{\F}\ctx(\F,\A)
&= -2\sigma_n^2\,\F\A\A^{\T} - 2\sigma_z^2\,\F, \label{eq:app:gradF-ctx}\\
\nabla_{\A}\ctx(\F,\A)
&= -2\sigma_n^2\,\F^{\T}\F\A. \label{eq:app:gradA-ctx}
\end{align}

\paragraph{Gradients of $\crx(\F,\A)$}
Using the product rule~\eqref{eq:app:dUV} and starting from $\crx=T\Prx-\mathrm{tr}(\A\Cov_w\A^{\T})$,
\begin{align}
d\,\mathrm{tr}(\A\Cov_w\A^{\T})
&= \mathrm{tr}((2\A\Cov_w)^{\T}d\A)
 +\mathrm{tr}(\A^{\T}\A\,d\Cov_w),
\label{eq:app:dArSigmaA}
\end{align}
where we used $\Cov_w=\Cov_w^{\T}$ and trace cyclicity.
Substituting $d\Cov_w$ from \eqref{eq:app:dSigma} and collecting the $d\F$ and $d\A$ parts gives
\begin{align}
\nabla_{\F} \crx
&=-2\sigma_n^2\,\A^{\T}\A(\I+\F\A)\A^{\T} - 2\sigma_z^2\,\A^{\T}\A\F,
\label{eq:app:gradF-crx}\\
\nabla_{\A} \crx
&=-2\A\Cov_w - 2\sigma_n^2\,\F^{\T}\A^{\T}\A(\I+\F\A).
\label{eq:app:gradA-crx}
\end{align}

Note that $\nabla_{\F}\|\A\g^\star\|^2=\mathbf 0$ and
\begin{equation}
\nabla_{\A}\|\A\g^\star\|^2 = 2\A\g^\star{\g^\star}^{\T}.
\label{eq:app:grad-Ag}
\end{equation}

By \eqref{eq:Lagrangian-g} and \eqref{eq:grad-V-envelope}, we have
\begin{align}
\nabla_{\F}\mathcal V
&= \nabla_{\F}\big(\g^{\star\T}\Cov_w^{-1}\g^\star\big)
+ \lambda_1^\star \nabla_{\F}\ctx
+ \lambda_2^\star \nabla_{\F}\crx,
\label{eq:app:assembleF}\\
\nabla_{\A}\mathcal V
&= \nabla_{\A}\big(\g^{\star\T}\Cov_w^{-1}\g^\star\big)
+ \lambda_1^\star \nabla_{\A}\ctx
+ \lambda_2^\star \nabla_{\A}(\crx-\|\A\g^\star\|^2).
\label{eq:app:assembleA}
\end{align}

Finally, substituting \eqref{eq:app:gradF-obj}, \eqref{eq:app:gradA-obj},
\eqref{eq:app:gradF-ctx}--\eqref{eq:app:gradA-ctx},
\eqref{eq:app:gradF-crx}--\eqref{eq:app:gradA-crx}, and \eqref{eq:app:grad-Ag}
into \eqref{eq:app:assembleF}--\eqref{eq:app:assembleA} yields the closed-form expressions
\eqref{eq:gradV-F} and \eqref{eq:gradV-A} stated in Theorem~\ref{thm:update-FA-envelop}.
\qed

\section{Proof of Lemma~\ref{cor:convergence}: Stationary Convergence}
\label{app:convergence}
This appendix first proves the monotone improvement property of Algorithm~\ref{alg:optimal-active-feedback-scheme}, namely that both the reduced objective value and the achieved SNR are nondecreasing, and then establishes stationary convergence.

Since $\F\A$ is strictly lower triangular,  $(\I+\F\A)$ has ones on its diagonal and hence nonsingular. Therefore,
\[
\Cov_w \succeq  \sigma_n^2 \sigma^2_{\min}(\I+\F\A)\I,
\]
where $\sigma_{\min}(\I+\F\A)\triangleq \sqrt{\lambda_{\min}((\I+\F\A)(\I+\F\A)^\T)}>0$, denotes the smallest singular value of $\I+\F\A$.

From the forward- and feedback-link power constraints in \eqref{eq:P1-active-feedback}, and using $\Cov_w \succeq \sigma_n^2 \sigma^2_{\min}(\I+\F\A) \I$, which implies $\tr(\A\Cov_w\A^\T) \geq \sigma_n^2 \sigma^2_{\min}(\I+\F\A)\;\tr(\A\A^\T)$, we obtain
\[
\| \F \|_{F}^2 \leq T\Ptx/\sigma_z^2, \quad \| \A \|_{F}^2 \leq T\Prx/(\sigma_n^2\sigma^2_{\min}(\I+\F\A)). 
\]
Therefore, the feasible set of causal matrix pairs $(\F,\A)$ is compact. Because Algorithm~\ref{alg:optimal-active-feedback-scheme} maintains feasibility, the sequence $\{(\F^{(k)}, \A^{(k)})\}$ lies in this compact set and thus has a convergent subsequence. 

\paragraph{Monotone Improvement of Algorithm~\ref{alg:optimal-active-feedback-scheme}}

Let $X^{(k)}\triangleq(\F^{(k)},\A^{(k)})$ and define the projected gradient direction
\[
d^{(k)} \triangleq \Pi_{\mathcal S}\big(\nabla \mathcal V(X^{(k)})\big).
\]
Since $\mathcal V$ has an $L$-Lipschitz gradient on the feasible set, the standard smoothness inequality yields,
for any $\eta>0$,
\[
\mathcal V(X^{(k)}+\eta d^{(k)})
\ge \mathcal V(X^{(k)}) + \eta \langle \nabla \mathcal V(X^{(k)}), d^{(k)}\rangle
-\frac{L}{2}\eta^2 \|d^{(k)}\|_F^2 .
\]
Because $\Pi_{\mathcal S}$ is an orthogonal projection, we have
\[
\langle \nabla \mathcal V(X^{(k)}), d^{(k)}\rangle
= \langle \Pi_{\mathcal S}(\nabla \mathcal V(X^{(k)})), d^{(k)}\rangle
= \|d^{(k)}\|_F^2 .
\]
Setting $\eta=\eta_k$ and using $\eta_k\le 1/L$ gives

\begin{equation}
\mathcal V(X^{(k+1)})-\mathcal V(X^{(k)})
\ge \eta_k\Big(1-\frac{L}{2}\eta_k\Big)\|d^{(k)}\|_F^2 \ge 0.
\label{eq:nondecreasing value function}
\end{equation}
Hence $\{\mathcal V(X^{(k)})\}$ is nondecreasing. Since $\mathcal V$ is continuous on a compact feasible set,
it is bounded above, and thus $\{\mathcal V(X^{(k)})\}$ is convergent.

Moreover, let
$\g^{(k)}=\g^\star(\F^{(k)},\A^{(k)})$ denote the exact inner maximizer at iteration $k$.
By the definition of $\mathcal V$ and the complementary slackness conditions
\begin{align}
\lambda_1^{(k)}\big(\|\g^{(k)}\|^2-c_{\rm fw}(\F^{(k)},\A^{(k)})\big)=0,
\\
\lambda_2^{(k)}\big(\|\A^{(k)}\g^{(k)}\|^2-c_{\rm fb}(\F^{(k)},\A^{(k)})\big)=0,
\end{align}
the constraint terms vanish in the Lagrangian, and therefore
\begin{align}
\mathcal V(\F^{(k)},\A^{(k)})
&=
\mathcal L\big(\g^{(k)};\lambda_1^{(k)},\lambda_2^{(k)}\big)\\
&=
(\g^{(k)})^\T \Cov_w^{-1}(\F^{(k)},\A^{(k)})\g^{(k)}
= \SNR^{(k)}.
\end{align}
Thus, \eqref{eq:nondecreasing value function} also implies
\begin{equation}\label{eq:nondecreasing SNR}
\SNR^{(k+1)} \ge \SNR^{(k)}, \qquad \forall k,
\end{equation}
which shows the monotone SNR improvement of Algorithm~\ref{alg:optimal-active-feedback-scheme}.

\paragraph{Stationary Convergence of Algorithm~\ref{alg:optimal-active-feedback-scheme}}

Summing the above inequality over $k=0,\dots,K-1$ yields
\[
\sum_{k=0}^{K-1}\frac{\eta_k}{2}\|d^{(k)}\|_F^2
\le \mathcal V(X^{(K)})-\mathcal V(X^{(0)})
< \infty,
\]
and letting $K\to\infty$ gives $\sum_{k=0}^{\infty}\eta_k\|d^{(k)}\|_F^2<\infty$.
Using $\eta_k\ge \eta_{\min}>0$, we conclude that
\[
\sum_{k=0}^{\infty}\|d^{(k)}\|_F^2 < \infty,
\qquad \text{and hence}\qquad
\|d^{(k)}\|_F \to 0,
\]
i.e.,
\[
\Big\|\Pi_{\mathcal S}\big(\nabla \mathcal V(\F^{(k)},\A^{(k)})\big)\Big\|_F \to 0.
\]

Finally, let $(\F^{(k_j)},\A^{(k_j)})\to(\F^\star,\A^\star)$ be any convergent subsequence guaranteed by compactness.
By continuity of $\Pi_{\mathcal S}(\nabla \mathcal V(\cdot))$, we obtain
\[
\Pi_{\mathcal S}\big(\nabla \mathcal V(\F^\star,\A^\star)\big)=(\O, \O),
\]
so every accumulation point is projected-stationary.
With $\g^\star=\g^\star(\F^\star,\A^\star)$ updated optimally by Theorem~\ref{thm:update-FA-envelop},
$(\g^\star,\F^\star,\A^\star)$ satisfies the first-order KKT stationarity conditions of~\eqref{eq:P1-active-feedback},
up to the sign symmetry. This completes the proof of Lemma~\ref{cor:convergence}\qed

\section{Proof of Theorem~\ref{thm:passive_GT_structure}: Geometric-Toeplitz Optimality in the Passive Case}
\label{app:passive_GT_proof}

\begin{mdframed}
\begin{cor}[Geometric structure in the passive case]
\label{cor:passive_GT_structure}
Consider the passive restriction $\A=\I$ and $\lambda_2^\star=0$. Let $(\g^\star,\lambda_1^\star)$ be the optimal primal--dual solution to the $\g$-subproblem
\eqref{eq:P2-g-subproblem} for a given feasible $\F$, and let $\S^\star:=\Cov_w^{-1}\g^\star(\Cov_w^{-1}\g^\star)^\T$ be the rank-one matrix introduced in
Theorem~\ref{thm:update-FA-envelop}.
If $\F^\star$ is a KKT (projected) stationary point of $\mathcal{V}$ over the causal subspace, i.e.,
\begin{equation}
\mathrm{StrictlyLower}\!\big(\nabla_{\F}\mathcal{V}(\F^\star,\A=\I)\big)=\O.
\label{eq:passive_projected_stationarity}
\end{equation}
Then there exists a positive geometric-Toeplitz (GT) ratio $\beta\in(0,1)$ such that $\F^\star$ and the corresponding $\g^\star$ admit a GT form:
\begin{equation}
g_t^\star=g_0(\pm\beta)^t,\qquad
F_{t,j}^\star=
\begin{cases}
0,& t\le j,\\
F_0(\pm\beta)^{t-j},& t>j,
\end{cases}
\label{eq:GT_form_cor}
\end{equation}
where $F_0 = -\left(\frac{\sigma_n^2}{\sigma_n^2+\sigma_z^2}\right) \frac{1-\beta^2}{\beta^2}$, and $\pm\beta$ form a twin pair of SNR/MMSE-equivalent solutions (the SNR/MSE performance depends only on $\beta^2$).
\end{cor}
\end{mdframed}

This corollary revisits Conjecture 1 in Section IV-A of the CL scheme~\cite{chance2011concatenated} and the conclusion of Section III-A in Agrawal \emph{et al.}~\cite{agrawal2011iteratively} for the optimal passive (uncoded) feedback case. We work under the passive restriction $\A=\I$ and $\lambda_2^\star=0$ under our notation.
Let $(\g^\star,\F^\star)$ denote a KKT-stationary solution and define the effective noise covariance
\[
\Cov_w(\F,\I)\triangleq \sigma_n^2(\I+\F)(\I+\F)^{\T}+\sigma_z^2\F\F^{\T}\succ \O .
\]

From Theorem~\ref{thm:update-FA-envelop}, setting $\A=\I$ and $\lambda_2^\star=0$ in \eqref{eq:gradV-F} yields
\begin{equation}
\nabla_{\F}\mathcal{V}(\F,\I)
=
-2\sigma_n^2 \S^\star(\I+\F)
-2\sigma_z^2 \S^\star \F
-2\lambda_1^\star(\sigma_n^2+\sigma_z^2)\F .
\label{eq:gradV_F_passive}
\end{equation}
Therefore, the projected stationarity condition \eqref{eq:passive_projected_stationarity},
i.e., $\mathrm{StrictlyLower}\!\big(\nabla_{\F}\mathcal{V}(\F^\star,\I)\big)=\O$,
is equivalent to the entrywise relations
\begin{equation}
\Big[\sigma_n^2 \S^\star(\I+\F^\star)+\sigma_z^2 \S^\star \F^\star
+\lambda_1^\star(\sigma_n^2+\sigma_z^2)\F^\star\Big]_{t,j}=0,
\quad \forall\, t>j .
\label{eq:stationary_entrywise}
\end{equation}

As we show next: (a) the KKT condition for the precoder $\g$ implies that $\S^\star$ is rank-one and is generated by the pre-whitened vector $\h^\star \triangleq \Cov_w^{-1/2}\g^\star$. (b) combining this rank-one structure with the projected KKT stationarity condition for $\F$ forces every strictly-lower column of $\F^\star$ to align with the corresponding tail of the pre-whitened vector $\h^\star$. (c) This column-alignment property then implies that $\F^\star$, $\g^\star$, and $\h^\star$ have a geometric-Toeplitz form.

\paragraph{Reducing to a pre-whitened formulation}
We introduce the \emph{pre-whitened} vector
\[
\h^\star \triangleq \Cov_w(\F^\star,\I)^{-1/2}\g^\star,
\]
so that $\|\h^\star\|^2={\g^\star}^{\T}\Cov_w^{-1}\g^\star=\SNR$.
By Theorem~\ref{thm:active-gq-eig}, the KKT condition with respect to $\g$ reduces to
$\Cov_w^{-1}\g^\star=\lambda_1^\star\g^\star$ and thus
\[
\Cov_w^{-1}\g^\star = \sqrt{\lambda_1^\star}\,\h^\star.
\]

Hence
\[
\S^\star=\big(\Cov_w^{-1}\g^\star\big)\big(\Cov_w^{-1}\g^\star\big)^{\T}
=\lambda_1^\star\,\h^\star(\h^\star)^{\T}
\]
is rank-one and completely determined by the pre-whitened $\h^\star$.

\paragraph{Column alignment of $\F^\star$ with the pre-whitened $\h^\star$}
Substituting $\S^\star=\lambda_1^\star \h^\star{\h^\star}^{\T}$ into \eqref{eq:stationary_entrywise}
and canceling the common factor $\lambda_1^\star>0$ gives, for all $t>j$,
\begin{equation}
\Big[\sigma_n^2 \h^\star{\h^\star}^{\T}(\I+\F^\star)+\sigma_z^2 \h^\star{\h^\star}^{\T}\F^\star
+(\sigma_n^2+\sigma_z^2)\F^\star\Big]_{t,j}=0.
\label{eq:stationary_entrywise_pre-whitened}
\end{equation}

Fix a column index $j\in\{0,\ldots,T-2\}$ and define its strictly-lower part
\[
\f_j \in \mathbb{R}^{T-j-1},\quad (\f_j)_\ell \triangleq F^\star_{j+1+\ell,\,j},\ \forall~ \ell=0,\ldots,T-j-2.
\]
Also define the tail vector and its energy
\[
\h^{(j)} \triangleq [h^\star_{j+1},\ldots,h^\star_{T-1}]^{\T},
\qquad
H_j \triangleq \|\h^{(j)}\|^2,
\]
and the scalar correlation
\[
c_j \triangleq (\h^{(j)})^{\T}\f_j = \sum_{t>j} h^\star_t\,F^\star_{t,j}.
\]
Evaluating~\eqref{eq:stationary_entrywise_pre-whitened} at entry $(t,j)$, using
$[\h^\star{\h^\star}^{\T}]_{t,j}=h^\star_t h^\star_j$
and
$[\h^\star{\h^\star}^{\T}\F^\star]_{t,j}=h^\star_t \sum_{i>j} h^\star_iF^\star_{i,j}
= h^\star_t c_j$,
we obtain for every $t>j$:
\[
(\sigma_n^2+\sigma_z^2)\F^\star_{t,j} + \sigma_n^2 h^\star_t h^\star_j
+(\sigma_n^2+\sigma_z^2) h^\star_t c_j = 0.
\]
In vector form this reads
\[
(\sigma_n^2+\sigma_z^2)\f_j + \big(\sigma_n^2 h^\star_j + (\sigma_n^2+\sigma_z^2)c_j\big)\,\h^{(j)}=\mathbf{0},
\]
so $\f_j$ must be aligned with $\h^{(j)}$:
\[
\f_j = -\Big(\frac{\sigma_n^2}{\sigma_n^2+\sigma_z^2}h^\star_j + c_j\Big)\h^{(j)}.
\]
Taking the inner product with $\h^{(j)}$ gives
\[
c_j = (\h^{(j)})^{\T}\f_j
= -\Big(\frac{\sigma_n^2}{\sigma_n^2+\sigma_z^2}h^\star_j + c_j\Big)H_j,
\]
hence
\begin{equation}
c_j = -\frac{\sigma_n^2}{\sigma_n^2+\sigma_z^2} \frac{H_j}{1+H_j} h^\star_j.
\label{eq:cj_closed}
\end{equation}
Substituting back yields the closed form
\begin{equation}
\f_j^\star
= -\frac{\sigma_n^2}{\sigma_n^2+\sigma_z^2}\frac{h^\star_j}{1+H_j}\,\h^{(j)},
\qquad j=0,\ldots,T-2,
\label{eq:passive_fj_closed}
\end{equation}
and equivalently, for each $t>j$,
\begin{equation}
F^\star_{t,j}
= -\frac{\sigma_n^2}{\sigma_n^2+\sigma_z^2}\frac{h^\star_j}{1+H_j}\,h^\star_t .
\label{eq:passive_F_entry_closed}
\end{equation}
This is the precise ``column-alignment'' property: each strictly-lower column is aligned with the tail of pre-whitened $\h$.
\paragraph{Geometric form of $\h^\star$ and geometric-Toeplitz $\F^\star$}
It remains to show that the pre-whitened vector $\h^\star$ is geometric.
Since $\Cov_w^{-1}\g^\star=\lambda_1^\star\g^\star$ and $\Cov_w\succ\O$,
$\g^\star$ is also an eigenvector of $\Cov_w$ with eigenvalue $1/\lambda_1^\star$.
Moreover, because $\h^\star=\Cov_w^{-1/2}\g^\star$ and $\g^\star$ is an eigenvector of the nonsingular symmetric matrix $\Cov_w^{-1}$,
we have $\h^\star=\sqrt{\lambda_1^\star}\,\g^\star$ by the Spectral Theorem. Hence
\begin{equation}
\Cov_w\,\h^\star = \frac{1}{\lambda_1^\star}\,\h^\star .
\end{equation}
Equivalently, the ratios $\big[\Cov_w\h^\star\big]_t/h_t^\star$ are equal to $1/\lambda^\star$ for all $t$.

Define $U_j \triangleq 1+H_j = 1+\sum_{t>j}(h_t^\star)^2$ and note that $(h_{j+1}^\star)^2=U_{j}-U_{j+1}$.
From \eqref{eq:stationary_entrywise}--\eqref{eq:passive_F_entry_closed}
, for each $t>j$ we can write
\begin{equation}
F_{t,j}^\star = -\phi_j\, h_t^\star,
\qquad
\phi_j \triangleq \frac{\sigma_n^2}{\sigma_n^2+\sigma_z^2}\frac{h_j^\star}{U_j}.
\label{eq:phi_j_closed}
\end{equation}
Then $(\F^\star)^{\T}\h^\star$ has entries
\begin{equation}
\big[(\F^\star)^{\T}\h^\star\big]_j
=\sum_{t>j}F_{t,j}^\star h_t^\star
= -\phi_j \sum_{t>j}(h_t^\star)^2
= -\phi_j H_j
= c_j,
\end{equation}
and hence $((\I+\F^\star)^{\T}\h^\star)_j=h_j^\star+c_j$.

Now expand $\Cov_w\h^\star
=\sigma_n^2(\I+\F^\star)(\I+\F^\star)^{\T}\h^\star+\sigma_z^2\F^\star{\F^\star}^{\T}\h^\star$.
Using the row structure in \eqref{eq:stationary_entrywise}--\eqref{eq:passive_F_entry_closed}, one verifies that for each $t\ge 0$,
\[
\frac{\big[\Cov_w\h^\star\big]_t}{h_t^\star}
=
\sigma_n^2\frac{h_t^\star+c_t}{h_t^\star}
-\sigma_n^2\sum_{j<t}\phi_j (h_j^\star+c_j)
-\sigma_z^2\sum_{j<t}\phi_j c_j .
\]
Taking the difference of the ratio expression $\big[\Cov_w \h^\star\big]_t/h_t^\star$ at $t+1$ and $t$
telescopes the cumulative sums, yielding
\begin{equation}
\sigma_n^2\Big(\frac{h_{t+1}^\star+c_{t+1}}{h_{t+1}^\star}-\frac{h_t^\star+c_t}{h_t^\star}\Big) \\
=\phi_t\big(\sigma_n^2(h_t^\star+c_t)+\sigma_z^2 c_t\big) 
\label{eq:telescoping_identity}
\end{equation}
Next, substitute $c_t=-(\sigma_n^2/(\sigma_n^2+\sigma_z^2))(H_t/U_t)h_t^\star$ in \eqref{eq:cj_closed} and
$\phi_t=(\sigma_n^2/(\sigma_n^2+\sigma_z^2))(h_t^\star/U_t)$ in \eqref{eq:phi_j_closed}.
A direct simplification gives
\[
\frac{h_t^\star+c_t}{h_t^\star}=\frac{\sigma_z^2}{\sigma_n^2+\sigma_z^2}+\frac{\sigma_n^2}{\sigma_n^2+\sigma_z^2}\frac{1}{U_t},
~~
\sigma_n^2(h_t^\star+c_t)+\sigma_z^2 c_t=\frac{\sigma_n^2}{U_t}\,h_t^\star.
\]
Substituting these identities into the telescoping result \eqref{eq:telescoping_identity} yields the recursion
\begin{equation}
\frac{1}{U_{t+1}}-\frac{1}{U_t}=\frac{(h_t^\star)^2}{U_t^2},
~\forall \; t=0,\ldots,T-2.
\end{equation}
Since $U_t-U_{t+1}=(h_{t+1}^\star)^2$, this is equivalent to
\[
\frac{(h_{t+1}^\star)^2}{U_tU_{t+1}}=\frac{(h_t^\star)^2}{U_t^2}
\quad\Longleftrightarrow\quad
\frac{U_{t+1}}{(h_{t+1}^\star)^2}=\frac{U_t}{(h_t^\star)^2}.
\]
Hence $U_t/(h_t^\star)^2$ is constant in $t$; i.e., there exists $R>0$ such that
\begin{equation}
U_t = R\,(h_t^\star)^2,\qquad \forall \; t=0,\ldots,T-1.
\end{equation}
Using $U_{t+1}=U_t-(h_{t+1}^\star)^2$ together with the previous display gives
$(R+1)(h_{t+1}^\star)^2=R(h_t^\star)^2$, so
\begin{equation}
h_{t+1}^\star = \beta \, h_t^\star,\qquad
\beta \triangleq \sqrt{\frac{R}{R+1}}\in(0,1),
\end{equation}
up to a global sign choice. Therefore $\h^\star$ is geometric: 
\[
h_t^\star = h_0^\star \beta^t.
\]

Finally, substitute $U_j=R(h_j^\star)^2$ into \eqref{eq:phi_j_closed}:
\begin{align}
F_{t,j}^\star
&= -\frac{\sigma_n^2}{\sigma_n^2+\sigma_z^2}\frac{1}{R}\cdot\frac{h_t^\star}{h_j^\star}
= F_0\,\beta^{\,t-j}, \qquad t>j, \nonumber \\
\intertext{where we used $h_t^\star/h_j^\star=\beta^{t-j}$, and}
F_0 &\triangleq -\frac{\sigma_n^2}{\sigma_n^2+\sigma_z^2}\frac{1}{R} = -\frac{\sigma_n^2}{\sigma_n^2+\sigma_z^2}\frac{1-\beta^2}{\beta^2} <0.
\end{align}

Thus $\F^\star$ is Toeplitz along strictly-lower diagonals with a geometric ratio $\beta$, i.e., $\F^\star$ is a geometric-Toeplitz matrix. Since $\g^\star = 1/\sqrt{\lambda_1^\star} \h^\star$, the optimal precoder $\g^\star$ inherits the same geometric ratio:
\[ 
g^\star_t = g_0\beta^t \quad \forall t=0,\ldots,T-1.
\]
This completes the proof. \qed. 

\section{Proof of Theorem~\ref{thm:unique-toeplitz-optimal-ratio}: Uniqueness of Geometric-Toeplitz ratio $\beta$}
\label{app:unique-toeplitz-optimal-ratio}

\begin{mdframed}
\begin{theorem} (Uniqueness of the Optimal Toeplitz Ratio for the Passive Feedback) \label{thm:unique-toeplitz-optimal-ratio}
Consider the passive linear feedback problem under $\A=\I$ and $\lambda_2^\star=0$.
For any $T\ge2$, $\Ptx>0$, $\sigma_n^2>0$, and $\sigma_z^2\geq0$, there exists a unique Toeplitz ratio $\beta^* \in (0,1)$ that maximizes $\mathrm{SNR}(\beta)$, i.e.,
\begin{equation}
\exists!~ \beta^* = \arg\max_{\beta >0} \mathrm{SNR}(\beta) ~\in (0,1).
\end{equation}
Moreover, the optimal Toeplitz ratios $\pm\beta^*$ correspond to the unique
positive root $\beta^* \in (0,1)$ of the polynomial equation $h(\beta)=0$, defined as:
\begin{equation}
\hspace{-0.7em}
h(\beta)\triangleq \Big[
\sigma_z^2
+T\Ptx\left(1+\frac{\sigma_z^2}{\sigma_n^2}\right)
+T\sigma_n^2
\Big]\beta^{2T}
-\sigma_n^2T\beta^{2T-2}
-\sigma_z^2.
\label{eq:h(beta)=0}
\end{equation}
In the case of $T=1$, which corresponds to the no-feedback scenario, the SNR is constant, $\mathrm{SNR}(\beta)=\rho/\sigma_n^2$, regardless of $\beta$.
\end{theorem}
\end{mdframed}

For $T=1$ (no feedback), the optimal design reduces to open-loop transmission and the received SNR is constant,
$\SNR(\beta)=\Ptx/\sigma_n^2$, independent of $\beta$.

Assume $T\ge 2$.
Under the passive restriction $\A=\I$ and $\lambda_2^\star=0$, Corollary~\ref{cor:passive_GT_structure} implies that any KKT-stationary optimal
solution admits the geometric-Toeplitz (GT) form
\begin{equation}
g_t = g_0(\pm\beta)^t,\qquad
F_{t,j} =
\begin{cases}
0, & t\le j,\\
F_0(\pm\beta)^{t-j}, & t>j,
\end{cases}
\label{eq:gt and F_tj}
\end{equation}
with some $\beta\in(0,1)$ and
\begin{equation}
F_0 = -\Big(\frac{\sigma_n^2}{\sigma_n^2+\sigma_z^2}\Big)\frac{1-\beta^2}{\beta^2}.
\label{eq:F0_beta}
\end{equation}
Since the received SNR depends only on $\beta^2$, it suffices to consider $\beta\in(0,1)$.

\paragraph{Closed-form $\SNR(\beta)$}
With $\A=\I$, the effective noise covariance is
\[
\Cov_w (\F,\A=\I) = \sigma_n^2(\I+\F)(\I+\F)^{\T} + \sigma_z^2\F\F^{\T}.
\]
Moreover, the forward-link average power constraint becomes
\begin{equation}
\|\g\|^2 + (\sigma_n^2+\sigma_z^2)\|\F\|_F^2 \le T\Ptx,
\label{eq:fw_power_passive}
\end{equation}
where we take $P_\theta=1$ (as in the main text) and denote the per-symbol average transmit-power budget by $\Ptx$.
Since the objective is increasing under positive scaling of $\g$, \eqref{eq:fw_power_passive} is tight at optimum.

For the geometric-Toeplitz matrix $\F$ in \eqref{eq:gt and F_tj} and \eqref{eq:F0_beta}, its Frobenius norm admits the closed form obtained by diagonal-wise summation and standard geometric/weighted-geometric series identities:
\begin{equation}
\begin{aligned}
\|\F\|_F^2
&= \sum_{d=1}^{T-1}(T-d)F_0^2\beta^{2d}
\\
&=  TF_0^2\sum_{d=1}^{T-1}\beta^{2d} - TF_0^2\sum_{d=1}^{T-1}d\beta^{2d}
\\
&= \Big(\frac{\sigma_n^2}{\sigma_n^2+\sigma_z^2}\Big)^2
\Big(\frac{T-1}{\beta^2}-T+\beta^{2T-2}\Big).
\end{aligned}
\label{eq:F_fro_beta}
\end{equation}

Hence the available precoder energy is
\begin{equation}
\begin{aligned}
\|g\|^2
&= T\Ptx - (\sigma_n^2+\sigma_z^2)\|\F\|_F^2\\
&= T\Ptx - \frac{\sigma_n^4}{\sigma_n^2+\sigma_z^2}
\Big(\frac{T-1}{\beta^2}-T+\beta^{2T-2}\Big).
\label{eq:g_norm_from_power}
\end{aligned}
\end{equation}

Define the geometric vector $\u(\beta)\triangleq[1,\beta,\ldots,\beta^{T-1}]^{\T}$.
For the GT choice \eqref{eq:F0_beta}, one can verify that $\u(\beta)$ is an eigenvector of $\Cov_w$ with eigenvalue
\begin{equation}
\lambda_1(\beta)=\frac{\sigma_n^2\big(\sigma_z^2+\sigma_n^2\beta^{2T-2}\big)}{\sigma_n^2+\sigma_z^2}.
\label{eq:mu_beta}
\end{equation}
In the passive case, the KKT condition for $\g$ implies $\g$ aligns with the dominant eigenvector of $\Cov_w^{-1}$,
hence we should take $\g \parallel \u(\beta)$, which yields
\[
\SNR(\beta)= \g^{\T}\Cov_w^{-1}\g = \frac{\|\g\|^2}{\lambda_1(\beta)}.
\]
Substituting \eqref{eq:g_norm_from_power}--\eqref{eq:mu_beta} and simplifying gives the closed-form SNR as an explicit function of $\beta$ and the system parameters $(T,\Ptx,\sigma_n^2,\sigma_z^2)$:
\begin{equation}
\SNR(\beta)
=
\frac{
\dfrac{T\Ptx}{\sigma_n^2}\big(1+\dfrac{\sigma_z^2}{\sigma_n^2}\big)\beta^2
-T(1-\beta^2)+(1-\beta^{2T})}
{
\dfrac{\sigma_z^2}{\sigma_n^2} \beta^2 +\beta^{2T}
}.
\label{eq:SNR_beta_closed}
\end{equation}

\begin{algorithm} [t!]     
\caption{Optimal Passive Feedback Scheme (Geometric Toeplitz Scheme)}  \label{alg:GT-passive-feedback-scheme}  

\KwIn{Blocklength $T$, forward noise variance $\sigma_n^2$, feedback noise variance $\sigma_z^2$, and average per-symbol transmit-power budget  $\Ptx$.}
\KwOut{Optimal geometric Toeplitz ratio $\beta^\star$, and the GT scheme design ($\mathbf g$, $\mathbf F$, $\mathbf q^\star_{\mathrm{LMMSE}}$).}

\BlankLine
Define the root-finding polynomial (Theorem~\ref{thm:unique-toeplitz-optimal-ratio})
\[
h(\beta)\triangleq \big[\sigma_z^2+T\Ptx\big(1+\tfrac{\sigma_z^2}{\sigma_n^2}\big)+T\sigma_n^2\big]\beta^{2T}
-\sigma_n^2T\beta^{2T-2}-\sigma_z^2.
\]
\\

\BlankLine
Find the unique root $\beta^* \in (0,1)$ such that $h(\beta^*)=0$ (e.g., using bisection or Newton's method).

\BlankLine
Compute the scalar prefactor $F_0$ of the transmitter encoding matrix $\F$ (Corollary~\ref{cor:passive_GT_structure}):
\[
F_0 \gets -\Big(\frac{\sigma_n^2}{\sigma_n^2+\sigma_z^2}\Big)\frac{1-\beta^2}{\beta^2}.
\]

Compute the Frobenius norm in closed form as a function of $\beta$ (Appendix~\ref{app:unique-toeplitz-optimal-ratio}):
\[
\|\mathbf F\|_F^2 \gets
\Big(\frac{\sigma_n^2}{\sigma_n^2+\sigma_z^2}\Big)^2
\Big(\frac{T-1}{\beta^2}-T+\beta^{2T-2}\Big).
\]

Recover the available precoder energy from the forward-link power constraint:
\[
\|\mathbf g\|^2 \gets T\Ptx-(\sigma_n^2+\sigma_z^2)\|\mathbf F\|_F^2.
\]

Compute the initial coefficient $g_0$ of the geometric precoder $\g$: 
\[
g_0 \gets \sqrt{\ \|\mathbf g\|^2\cdot\frac{1-\beta^2}{1-\beta^{2T}}\ }\ >0.
\]

\BlankLine
Construct $\mathbf g$ and $\mathbf F$ from $(g_0,F_0,\beta)$ using the geometric Toeplitz structure in \eqref{eq:GT_form_cor}, and compute
$\mathbf q=\mathbf q_{\mathrm{LMMSE}}(\mathbf g,\mathbf F)$ as in \eqref{eq:q_LMMSE}.

\BlankLine
\BlankLine
\Return $\beta^\star,\ \g,\ \mathbf F,\ \mathbf q_{\mathrm{LMMSE}}$
\end{algorithm}

\smallskip
\paragraph{$\SNR(\beta)$ maximizer condition $\iff h(\beta)=0$}
Differentiating \eqref{eq:SNR_beta_closed} with respect to $\beta$ yields
\begin{equation}
\frac{d}{d\beta}\SNR(\beta)
=
-\frac{(T-1)\sigma_n^2\cdot 2\beta}{\big(\sigma_z^2\beta^2+\sigma_n^2\beta^{2T}\big)^2}\,h(\beta),
\label{eq:dSNR_dbeta_hbeta}
\end{equation}
where $h(\beta)$ is exactly the polynomial defined in \eqref{eq:h(beta)=0}:
\[
h(\beta)\triangleq \Big[
\sigma_z^2
+T\Ptx\left(1+\frac{\sigma_z^2}{\sigma_n^2}\right)
+T\sigma_n^2
\Big]\beta^{2T}
-\sigma_n^2T\beta^{2T-2}
-\sigma_z^2.
\]
Since the prefactor in \eqref{eq:dSNR_dbeta_hbeta} is strictly negative for $\beta\in(0,1)$, we have
$\SNR'(\beta)=0$ if and only if $h(\beta)=0$, and
\[
\mathrm{sign}\big(\SNR'(\beta)\big) = -\,\mathrm{sign}\big(h(\beta)\big).
\]

\paragraph{$h(\beta)$ has a unique root in $(0,1)$}
Let $r\triangleq \beta^2\in(0,1)$ and define $\tilde h(r)\triangleq h(\beta)$, i.e.,
\begin{align}
\tilde h(r)=A r^T - \sigma_n^2T r^{T-1}-\sigma_z^2,  \nonumber\\
\intertext{where} A\triangleq \sigma_z^2+T\rho(1+\dfrac{\sigma_z^2}{\sigma_n^2})+T\sigma_n^2 \;>\;0.
\end{align}
Then
\begin{equation}
\begin{aligned}
&\tilde h(0)=-\sigma_z^2<0,\\
&\tilde h(1)=A-\sigma_n^2T-\sigma_z^2 = T\rho(1+\dfrac{\sigma_z^2}{\sigma_n^2})>0,
\end{aligned}
\end{equation}
so at least one root lies in $(0,1)$.
Moreover,
\[
\tilde h'(r) = Tr^{T-2}\big(Ar-\sigma_n^2(T-1)\big),
\]
which is negative for $0<r<r_0$ and positive for $r>r_0$, where
$r_0\triangleq \sigma_n^2(T-1)/A\in(0,1)$.
Hence $\tilde h(r)$ is strictly decreasing on $(0,r_0)$ and strictly increasing on $(r_0,1)$, implying that
$\tilde h$ (and thus $h$) can cross zero at most once on $(0,1)$.
Therefore, $h(\beta)=0$ has a unique root $\beta^\star\in(0,1)$.

The unique root $\beta^\star$ can therefore be found by bisection.
Since each evaluation of $h(\beta)$ requires $O(\log T)$ arithmetic
operations via repeated squaring~\cite{knuth1998seminumerical},
the resulting root search has $O(\log T)$ arithmetic complexity
for fixed accuracy.

Finally, since $\SNR'(\beta)$ changes sign exactly once on $(0,1)$ by \eqref{eq:dSNR_dbeta_hbeta}, the SNR in
\eqref{eq:SNR_beta_closed} is strictly increasing for $\beta<\beta^\star$ and strictly decreasing for $\beta>\beta^\star$.
Thus $\beta^\star$ is the unique maximizer of $\SNR(\beta)$ on $(0,1)$, and the twin ratios $\pm\beta^\star$ are
SNR/MMSE-equivalent. 

We are now ready to state the globally optimal passive (uncoded) feedback scheme. By Corollary~\ref{cor:passive_GT_structure} and Theorem~\ref{thm:unique-toeplitz-optimal-ratio}, the optimal passive design is completely determined by the unique geometric-Toeplitz ratio $\beta^\star \in (0,1)$. Algorithm~\ref{alg:GT-passive-feedback-scheme} provides a closed-form realization of the globally optimal passive specialization of Algorithm~\ref{alg:optimal-active-feedback-scheme}. 
\qed

\section{Proof of SNR/MSE Converse Bounds }
\label{app:converse bounds}

The active SNR upper bound \eqref{eq:SNR_active_UB} is a direct restatement of the Elias--Butman SNR upper bound for the noisy feedback case under our notation. Elias derived it via a circuit-theoretic argument (see ``Noisy Feedback, General Case'' discussion in~\cite{elias1967networks}), and Butman re-derived it in his matrix formulation (see Sec.~IV. ``Noisy Feedback'' in~\cite{butman1969general}). 

We begin with a simple decoupling bound on $\Cov_w^{-1}(\F,\A)$ in \eqref{eq:P1-active-feedback}, which allows us to upper bound the received SNR by separate forward- and feedback-link terms.
\begin{mdframed} 
\begin{lemma} [Decoupled Inverse-Covariance Bound]
\label{lem:decoupled-inverse-covariance-bound}
For any strictly lower-triangular $\F$ and lower-triangular $\A$, let
\[
\Cov_w(\F,\A)=\sigma_n^2(\I+\F\A)(\I+\F\A)^{\T}+\sigma_z^2\F\F^{\T}.
\]
Then
\begin{equation}
\Cov_w^{-1}(\F,\A)\ \preceq\ \frac{1}{\sigma_n^2}\I+\frac{1}{\sigma_z^2}\A^{\T}\A,
\label{eq:Covw_inv_UB}
\end{equation}
equivalently,
\begin{equation}
\SNR(\g,\F,\A) \triangleq \g^{\T}\Cov_w^{-1}(\F,\A)\g
 \le \frac{\|\g\|^2}{\sigma_n^2}+\frac{\|\A\g\|^2}{\sigma_z^2},~\forall\,\g.
\label{eq:SNR_g_Ag_UB}
\end{equation}
\end{lemma}    
\end{mdframed}

\begin{proof}
Let $\M:=\I+\F\A$ and define $\B:=[\,\sigma_n\M\ \ \sigma_z\F\,]$, so that
$\Cov_w=\B\B^{\T} \succ 0$.
For any $\g$, the quadratic form admits the variational representation
\begin{equation}
\g^{\T}\Cov_w^{-1}\g
=\min_{\a,\b:\ \sigma_n\M\a+\sigma_z\F\b=\g}\big(\|\a\|^2+\|\b\|^2\big).
\label{eq:variational_rep_SNR}
\end{equation}
Choose the feasible pair $\a=\g/\sigma_n$ and $\b=-\A\g/\sigma_z$.
Then $\sigma_n\M\a+\sigma_z\F\b=\M\g-\F\A\g=\g$, hence \eqref{eq:variational_rep_SNR} yields
\[
\g^{\T}\Cov_w^{-1}\g
\le \frac{\|\g\|^2}{\sigma_n^2}+\frac{\|\A\g\|^2}{\sigma_z^2},
\]
which proves \eqref{eq:SNR_g_Ag_UB} and thus \eqref{eq:Covw_inv_UB}.   
\end{proof}

Applying the above inequality in Lemma~\ref{lem:decoupled-inverse-covariance-bound} to the active feedback problem gives,
\[
\SNR_{\rm active}
={\g^\star}^{\T}\Cov_w^{-1}(\F^\star,\A^\star)\g^\star
\le \frac{\|\g^\star\|^2}{\sigma_n^2}+\frac{\|\A^\star\g^\star\|^2}{\sigma_z^2}.
\]
Since $(\F^\star,\A^\star)$ is feasible, \eqref{eq:P1-active-feedback} power constraints imply $\|\g^\star\|^2\le T\Ptx$ and $\|\A^\star\g^\star\|^2\le T\Prx$. Hence,
\[
\SNR_{\rm active}\le \frac{T\Ptx}{\sigma_n^2}+\frac{T\Prx}{\sigma_z^2}=\SNR_{\fw} + \SNR_{\fb}
\]
, proving the Elias--Butman bound \eqref{eq:SNR_active_UB}.

For passive feedback, Kim \textit{et al.}~\cite{kim2011error}
derived the corresponding error-exponent converse, while
Chance and Love~\cite{chance2011concatenated} gave the received-SNR
upper bound for linear feedback. For the passive restriction $\A=\alpha\I$, the feedback encoder transmits a scaled raw observation
$v_t=\alpha y_t=\alpha(x_t+n_t)$, which is precisely the uncoded noisy-output-feedback model
considered in CL. Under the feedback power constraint,
\[
\mathbb{E}[v_t^2]=\alpha^2\mathbb{E}[y_t^2]
=\alpha^2\big(\mathbb{E}[x_t^2]+\sigma_n^2\big)\le \Prx,
\]
so $\alpha^2\le \Prx/(\mathbb{E}[x_t^2]+\sigma_n^2)$. The component of $v_t$ correlated with the
information-bearing signal $x_t$ has power $\alpha^2\mathbb{E}[x_t^2]$, hence
\[
\alpha^2\mathbb{E}[x_t^2]\le \Prx \frac{\mathbb{E}[x_t^2]}{\mathbb{E}[x_t^2]+\sigma_n^2}
\le \frac{\Ptx}{\Ptx+\sigma_n^2}\Prx,
\]
where we used $\mathbb{E}[x_t^2]\le \Ptx$. Therefore, at most the fraction
$\frac{\Ptx}{\Ptx +\sigma_n^2}$ of the feedback SNR budget can be devoted to the
information-bearing part, yielding Chance--Love bound~\eqref{eq:SNR_passive_UB}. In particular, the noisy passive output-feedback model in CL (and the passive bound in Lemma~3 of~\cite{chance2011concatenated}) corresponds to the unit-gain special case $\alpha=1$ (i.e., $\A=\I$), which can be viewed as a simplified setting in which no explicit feedback-link power constraint is imposed. \qed

\section{Proof of Theorem~\ref{thm:achievability}:
Achievability of the Converse Bound}
\label{app:achievability}
The achievability proof uses the asymptotically optimal one-time feedback construction of~\cite{tung2026achieving} as a feasible initialization for Algorithm~\ref{alg:optimal-active-feedback-scheme}. We thank the authors of~\cite{tung2026achieving} for sharing the technical details of the one-time feedback approach and its asymptotic analysis with us. In this appendix, we establish first-order asymptotic optimality for active feedback
through constant-order additive converse-gap bounds, and then prove
the convergence of the associated Lagrange multipliers.
Throughout this appendix, $P_\theta=1$ and
$\Ptx,\Prx,\sigma_n^2,\sigma_z^2>0$ are fixed as $T$ increases.

\subsection{Achievability of the Elias--Butman SNR Bound}

Define the Elias--Butman converse and the relative link quality by
\begin{align}
B_T^{\rm EB}
&\triangleq
\frac{T\Ptx}{\sigma_n^2}
+\frac{T\Prx}{\sigma_z^2},
\\
\rho
&\triangleq
\frac{\SNR_{\fb}}{\SNR_{\fw}}
=
\frac{\Prx\sigma_n^2}{\Ptx\sigma_z^2}>0.
\label{eq:appG:link-ratio}
\end{align}
For any feasible causal linear scheme, define its additive gap as
\begin{equation}
\Gap(T;\g,\F,\A)
\triangleq
B_T^{\rm EB}-\SNR(\g,\F,\A).
\label{eq:appG:pointwise-gap}
\end{equation}
The optimal additive gap over all feasible schemes is
\begin{equation}
\Gap^\star(T)
\triangleq
B_T^{\rm EB}
-\max_{(\g,\F,\A)\ {\rm feasible}}
\SNR(\g,\F,\A).
\label{eq:appG:optimal-gap}
\end{equation}

\subsubsection{Finite-Blocklength Nonattainment}

Define the residual link-power budgets after subtracting the
message-bearing terms:
\begin{equation}
\Delta_{\rm fw}
\triangleq T\Ptx-\|\g\|^2,
\qquad
\Delta_{\rm fb}
\triangleq T\Prx-\|\A\g\|^2.
\label{eq:appG:residual-link-powers}
\end{equation}
By Lemma~\ref{lem:decoupled-inverse-covariance-bound}
and the power constraints,
\begin{equation}
\Gap(T;\g,\F,\A)
\ge
\frac{\Delta_{\rm fw}}{\sigma_n^2}
+\frac{\Delta_{\rm fb}}{\sigma_z^2}
\ge0.
\label{eq:appG:residual-gap}
\end{equation}

Suppose that $\Gap(T;\g,\F,\A)=0$ for a finite $T$.
Then \eqref{eq:appG:residual-gap} forces
$\Delta_{\rm fw}=\Delta_{\rm fb}=0$.
The forward-link power constraint gives
\begin{equation}
\sigma_n^2\|\F\A\|_F^2
+\sigma_z^2\|\F\|_F^2=0,
\end{equation}
and therefore $\F=\O$. Consequently,
$\Cov_w=\sigma_n^2\I$.
On the other hand, $\Delta_{\rm fb}=0$ and the feedback-link
power constraint imply
\begin{equation}
0\ge
\tr(\A\Cov_w\A^{\T})
=
\sigma_n^2\|\A\|_F^2,
\end{equation}
so $\A=\O$. This contradicts
$\|\A\g\|^2=T\Prx>0$. Hence
\begin{equation}
\boxed{
\Gap(T;\g,\F,\A)>0
\quad\text{for every finite }T.
}
\label{eq:appG:finite-nonattainment}
\end{equation}

\subsubsection{Asymptotic Gap Lower Bound}

We next strengthen the preceding nonattainment result to a
quantitative lower bound.
Set $x\triangleq\|\F\A\|_F$.
The forward-link power constraint and
\eqref{eq:appG:residual-gap} imply
\begin{equation}
\Gap(T;\g,\F,\A)
\ge
x^2+\frac{\Delta_{\rm fb}}{\sigma_z^2}.
\label{eq:appG:master-lower}
\end{equation}
If $x\ge1$, the gap is at least one.
Suppose therefore that $0\le x<1$.
Here, $\|\cdot\|_2$ denotes the Euclidean norm for vectors
and the induced spectral norm for matrices, whereas
$\|\cdot\|_F$ denotes the Frobenius norm.
For a matrix $\mathbf M$ with singular values $\sigma_i(\mathbf M)$,
\[
\begin{aligned}
\|\mathbf M\|_2
&=\max_{\|\mathbf v\|_2=1}\|\mathbf M\mathbf v\|_2
=\sigma_{\max}(\mathbf M),\\
\|\mathbf M\|_F
&=\left(\sum_{i,j}|M_{ij}|^2\right)^{1/2}
=\left(\sum_i\sigma_i(\mathbf M)^2\right)^{1/2}.
\end{aligned}
\]
Since $\|\F\A\|_2\le\|\F\A\|_F$,
\begin{equation}
\sigma_{\min}(\I+\F\A)
\ge1-\|\F\A\|_2
\ge1-\|\F\A\|_F
=1-x.
\end{equation}
Hence, the effective noise covariance satisfies
\begin{equation}
\Cov_w\succeq\sigma_n^2(1-x)^2\I.
\end{equation}
Using
$\|\A\g\|^2\le\|\A\|_F^2\|\g\|^2
\le T\Ptx\|\A\|_F^2$, the feedback constraint yields
\begin{align}
\Delta_{\rm fb}
&\ge
\tr(\A\Cov_w\A^{\T})
\nonumber\\
&\ge
\sigma_n^2(1-x)^2\|\A\|_F^2
\nonumber\\
&\ge
\frac{\sigma_n^2(1-x)^2}{T\Ptx}
\|\A\g\|^2
\nonumber\\
&=
\frac{\sigma_n^2(1-x)^2}{T\Ptx}
(T\Prx-\Delta_{\rm fb}).
\end{align}
Define
\[
\epsilon_T\triangleq\frac{\sigma_n^2}{T\Ptx}.
\]
Solving the preceding scalar inequality gives
\begin{equation}
\frac{\Delta_{\rm fb}}{\sigma_z^2}
\ge
\frac{\rho(1-x)^2}
{1+\epsilon_T(1-x)^2}.
\end{equation}
Consequently, for $0\le x<1$,
\begin{align}
\Gap(T;\g,\F,\A)
&\ge
x^2+\frac{\rho}{1+\epsilon_T}(1-x)^2
\nonumber\\
&\ge
\frac{\rho}{1+\rho+\epsilon_T}.
\end{align}
The last inequality follows by minimizing the quadratic
expression over $x$.

Since the resulting bound is strictly less than one, it also
applies when $x\ge1$. Thus every feasible scheme satisfies
\begin{equation}
\Gap(T;\g,\F,\A)
\ge
L_T
\triangleq
\frac{\rho}{1+\rho+\epsilon_T}>0.
\label{eq:appG:finite-lower}
\end{equation}
Minimizing over feasible schemes and taking $T\to\infty$
therefore yields
\begin{equation}
\boxed{
\liminf_{T\to\infty}\Gap^\star(T)
\ge
\frac{\rho}{1+\rho}>0.
}
\label{eq:appG:asymptotic-lower}
\end{equation}

\subsubsection{Constant Upper Bound on the Asymptotic Gap}
Following the idea in~\cite{tung2026achieving}, we consider the following two tap feasible construction. Let $\mathbf e_1,\mathbf e_2\in\mathbb R^T$ denote the first two coordinate vectors. Consider the feasible construction
\begin{equation}
\begin{aligned}
\A_{\mathrm{2tap}}
&=a\mathbf e_1\mathbf e_1^{\T},\\
\F_{\mathrm{2tap}}
&=c\mathbf e_2\mathbf e_1^{\T},\\
\g_{\mathrm{2tap}}
&=u_T(\mathbf e_1-\mathbf e_2),
\end{aligned}
\label{eq:appG:two-construction}
\end{equation}
where
\begin{equation}
a=\sqrt{\frac{2\Prx}{\Ptx}},
\qquad
c=\frac{\sigma_n^2}{\sigma_z^2}a.
\end{equation}
The matrices $\A_{\mathrm{2tap}}$ and $\F_{\mathrm{2tap}}$ are lower
triangular and strictly lower triangular, respectively.
Write
\[
\Cov_{w,\mathrm{2tap}}
\triangleq
\Cov_w(\F_{\mathrm{2tap}},\A_{\mathrm{2tap}}).
\]

Define
\begin{align}
B_{\rm f}
&\triangleq
c^2(\sigma_n^2a^2+\sigma_z^2)
\nonumber\\
&=
2\sigma_n^2\rho(1+2\rho),
\\
m
&\triangleq
\max\left\{\sigma_n^2,\frac{B_{\rm f}}{2}\right\},
\\
u_T^2
&\triangleq
\frac{T\Ptx}{2}-m.
\end{align}
For every integer
\[
T\ge T_0
\triangleq
\max\left\{
2,\left\lfloor\frac{2m}{\Ptx}\right\rfloor+1
\right\},
\]
we choose the positive square root $u_T>0$.

The forward-link power satisfies
\begin{align}
&\|\g_{\mathrm{2tap}}\|^2
+\sigma_n^2\|\F_{\mathrm{2tap}}\A_{\mathrm{2tap}}\|_F^2
+\sigma_z^2\|\F_{\mathrm{2tap}}\|_F^2
\nonumber\\
&\qquad=
2u_T^2+B_{\rm f}
\nonumber\\
&\qquad=
T\Ptx-2m+B_{\rm f}
\le T\Ptx.
\end{align}
Moreover, the first diagonal entry of $\Cov_{w,\mathrm{2tap}}$
is $\sigma_n^2$. The feedback-link power is therefore
\begin{align}
&\|\A_{\mathrm{2tap}}\g_{\mathrm{2tap}}\|^2
+\tr\!\left(
\A_{\mathrm{2tap}}\Cov_{w,\mathrm{2tap}}\A_{\mathrm{2tap}}^{\T}
\right)
\nonumber\\
&\qquad=
a^2(u_T^2+\sigma_n^2)
\nonumber\\
&\qquad=
T\Prx+a^2(\sigma_n^2-m)
\le T\Prx.
\end{align}
Thus the construction is feasible under both block-average
power constraints.

Since $\sigma_z^2c^2=\sigma_n^2ac$, the nontrivial leading
$2\times2$ block of $\Cov_{w,\mathrm{2tap}}$ is
\begin{equation}
\Cov_{w,\mathrm{2tap}}^{(2)}
=
\sigma_n^2
\begin{bmatrix}
1&ac\\
ac&1+ac+a^2c^2
\end{bmatrix}.
\end{equation}
Direct inversion gives
\begin{equation}
\begin{aligned}
&\left(\Cov_{w,\mathrm{2tap}}^{(2)}\right)^{-1}
=
\frac{1}{\sigma_n^2(\sigma_n^2+\sigma_z^2c^2)}
\\
&\quad\times
\begin{bmatrix}
\sigma_n^2(1+a^2c^2)+\sigma_z^2c^2 & -\sigma_n^2ac\\
-\sigma_n^2ac & \sigma_n^2
\end{bmatrix}.
\end{aligned}
\end{equation}
Since the nonzero part of $\g_{\mathrm{2tap}}$ is
$u_T[\,1\ -1\,]^{\T}$ and the remaining coordinates are
uncoupled from this block, direct substitution gives
\begin{align}
\SNR_{\mathrm{2tap}}(T)
&\triangleq
\g_{\mathrm{2tap}}^{\T}
\Cov_{w,\mathrm{2tap}}^{-1}
\g_{\mathrm{2tap}}
\nonumber\\
&=
u_T^2
\begin{bmatrix}1&-1\end{bmatrix}
\left(\Cov_{w,\mathrm{2tap}}^{(2)}\right)^{-1}
\begin{bmatrix}1\\-1\end{bmatrix}
\nonumber\\
&=
u_T^2
\frac{\sigma_n^2(2+2ac+a^2c^2)+\sigma_z^2c^2}
{\sigma_n^2(\sigma_n^2+\sigma_z^2c^2)}
\nonumber\\
&=
u_T^2
\frac{\sigma_n^2(1+ac)(2+ac)}
{\sigma_n^4(1+ac)}
\nonumber\\
&=
\frac{u_T^2(2+ac)}{\sigma_n^2}
\nonumber\\
&=
\frac{2u_T^2}{\sigma_n^2}
+\frac{2\Prx u_T^2}{\Ptx\sigma_z^2}
\nonumber\\
&=
\frac{T\Ptx}{\sigma_n^2}
+\frac{T\Prx}{\sigma_z^2}
-
2m\left(\frac{1}{\sigma_n^2}
+\frac{\Prx}{\Ptx\sigma_z^2}\right)
\nonumber\\
&=
B_T^{\rm EB}-U_{\mathrm{2tap}},
\label{eq:appG:two-snr}
\end{align}
where we used $\sigma_z^2c^2=\sigma_n^2ac$,
$ac=2\rho$, and $u_T^2=T\Ptx/2-m$, and defined
\begin{align}
U_{\mathrm{2tap}}
&\triangleq
\frac{2m}{\sigma_n^2}(1+\rho)
\nonumber\\
&=
2(1+\rho)\max\{1,\rho(1+2\rho)\}.
\label{eq:appG:two-gap-constant}
\end{align}
This upper-bound constant depends only on $\rho$ and is independent
of $T$. Therefore,
\begin{equation}
\Gap^\star(T)\le U_{\mathrm{2tap}},
\qquad T\ge T_0,
\label{eq:appG:optimal-upper}
\end{equation}
and consequently
\begin{equation}
\boxed{
\limsup_{T\to\infty}\Gap^\star(T)
\le
2(1+\rho)\max\{1,\rho(1+2\rho)\}.
}
\end{equation}

\subsubsection{Implication for Algorithm~\ref{alg:optimal-active-feedback-scheme}}

Initialize the algorithm with
\[
(\F^{(0)},\A^{(0)})
=
(\F_{\mathrm{2tap}},\A_{\mathrm{2tap}}).
\]
Because $\g_{\mathrm{2tap}}$ is feasible for the inner
$\g$-subproblem at this pair, exact inner optimization gives
\[
\SNR_{\mathrm{Alg.1}}^{(0)}(T)
\ge
\SNR_{\mathrm{2tap}}(T).
\]
The outer iterations preserve feasibility and monotonically
improve the achieved SNR. Consequently, after any number
of such iterations,
\begin{equation}
\SNR_{\mathrm{Alg.1}}(T)
\ge
\SNR_{\mathrm{Alg.1}}^{(0)}(T)
\ge
B_T^{\rm EB}-U_{\mathrm{2tap}}.
\end{equation}
Define the additive gap of the algorithm as
\[
\Gap(T)
\triangleq
B_T^{\rm EB}-\SNR_{\mathrm{Alg.1}}(T).
\]
Combining feasibility, \eqref{eq:appG:finite-lower}, and the
preceding comparison yields
\begin{equation}
0<L_T
\le
\Gap^\star(T)
\le
\Gap(T)
\le
U_{\mathrm{2tap}},
\qquad T\ge T_0.
\label{eq:appG:algorithm-gap-bounds}
\end{equation}
Thus
\begin{equation}
\boxed{
\Gap^\star(T)=\Theta(1),
\qquad
\Gap(T)=\Theta(1).
}
\label{eq:appG:constant-gap}
\end{equation}
Since $B_T^{\rm EB}=\Theta(T)$,
\begin{equation}
0\le
1-\frac{\SNR_{\mathrm{Alg.1}}(T)}{B_T^{\rm EB}}
\le
\frac{U_{\mathrm{2tap}}}{B_T^{\rm EB}}
\longrightarrow0.
\label{eq:appG:normalized-optimality}
\end{equation}
This establishes first-order asymptotic attainment of the
Elias--Butman converse, although the additive gap remains
bounded away from zero.

\subsection{Convergence of Lagrangian Multipliers}

For each $T\ge T_0$, let $(\F^\star,\A^\star)$ be an
accumulation point generated by
Algorithm~\ref{alg:optimal-active-feedback-scheme},
initialized as above, that is projected-stationary on the
causal constraint set.
For fixed $(\F^\star,\A^\star)$, let $\g^\star$ be the optimal
solution of the inner $\g$-subproblem, with associated KKT
multipliers $(\lambda_1^\star,\lambda_2^\star)$.
Their dependence on $T$ is suppressed for brevity, and we write
$\Cov_w=\Cov_w(\F^\star,\A^\star)$.

From the inner KKT condition \eqref{eq:KKT-stationary-g}
and Theorem~\ref{thm:active-gq-eig}, we have
\begin{equation}
\begin{aligned}
\SNR_{\mathrm{Alg.1}}
&\triangleq
{\g^\star}^{\T}\Cov_w^{-1}\g^\star
\\
&=
\lambda_1^\star\|\g^\star\|^2
+\lambda_2^\star\|\A^\star\g^\star\|^2.
\end{aligned}
\label{eq:appG:snr-multiplier-identity}
\end{equation}
Moreover, the converse bound and
Lemma~\ref{lem:decoupled-inverse-covariance-bound} yield
\[
\SNR_{\mathrm{Alg.1}}
\le
\frac{T\Ptx}{\sigma_n^2}
+\frac{T\Prx}{\sigma_z^2}
=
\SNR_{\fw}+\SNR_{\fb}.
\]

Recall the residual power budgets in
\eqref{eq:residual forward-power budgets} and
\eqref{eq:residual feedback-power budgets}:
\begin{align*}
\ctx(\F^\star,\A^\star)
&=
T\Ptx
-\sigma_n^2\|\F^\star\A^\star\|_F^2
-\sigma_z^2\|\F^\star\|_F^2,
\\
\crx(\F^\star,\A^\star)
&=
T\Prx
-\tr\!\left(
\A^\star\Cov_w{\A^\star}^{\T}
\right).
\end{align*}
By complementary slackness for the inner problem,
\begin{align*}
\lambda_1^\star\|\g^\star\|^2
&=
\lambda_1^\star\ctx(\F^\star,\A^\star),
\\
\lambda_2^\star\|\A^\star\g^\star\|^2
&=
\lambda_2^\star\crx(\F^\star,\A^\star).
\end{align*}

At this stationary point, the residual gaps defined in
\eqref{eq:appG:residual-link-powers} satisfy
\begin{align*}
\Delta_{\rm fw}
&=
T\Ptx-\|\g^\star\|^2
\\
&\ge
T\Ptx-\ctx(\F^\star,\A^\star)
\\
&=
\sigma_n^2\|\F^\star\A^\star\|_F^2
+\sigma_z^2\|\F^\star\|_F^2,
\\
\Delta_{\rm fb}
&=
T\Prx-\|\A^\star\g^\star\|^2
\\
&\ge
T\Prx-\crx(\F^\star,\A^\star)
\\
&=
\tr\!\left(
\A^\star\Cov_w{\A^\star}^{\T}
\right).
\end{align*}

Using these definitions, the difference between the converse
and the SNR achieved by the algorithm can be expressed as
\begin{align}
\Gap(T)
&=
(\SNR_{\fw}+\SNR_{\fb})
-\SNR_{\mathrm{Alg.1}}
\nonumber\\
&=
\left(\frac1{\sigma_n^2}-\lambda_1^\star\right)T\Ptx
\nonumber\\
&\quad+
\left(\frac1{\sigma_z^2}-\lambda_2^\star\right)T\Prx
\nonumber\\
&\quad+
\lambda_1^\star\Delta_{\rm fw}
+\lambda_2^\star\Delta_{\rm fb}
\ge0.
\label{eq:gap_identity_appB}
\end{align}

As established above, $\Gap(T)=\Theta(1)$ and hence
$\Gap(T)/T\to0$.
By \eqref{eq:appG:residual-gap},
$\Delta_{\rm fw},\Delta_{\rm fb}=O(1)$.
Consequently,
\begin{align}
\|\g^\star\|^2
&=
T\Ptx-\Delta_{\rm fw}
=
T\Ptx+O(1),
\\
\|\A^\star\g^\star\|^2
&=
T\Prx-\Delta_{\rm fb}
=
T\Prx+O(1).
\label{eq:appG:message-power-asymptotics}
\end{align}
Since $\lambda_1^\star,\lambda_2^\star\ge0$,
\eqref{eq:appG:snr-multiplier-identity} then implies
$\lambda_1^\star,\lambda_2^\star=O(1)$.

Let
\[
\mathbf d^\star
\triangleq
\Cov_w^{-1}\g^\star.
\]
Recall the minimum-norm representation in
\eqref{eq:variational_rep_SNR}, used in the proof of
Lemma~\ref{lem:decoupled-inverse-covariance-bound}.
The feasible pair chosen there is
\[
\mathbf a_0=\frac{\g^\star}{\sigma_n},
\qquad
\mathbf b_0=-\frac{\A^\star\g^\star}{\sigma_z}.
\]
The corresponding minimum-norm pair is
\[
\begin{aligned}
\mathbf a_{\mathrm{min}}
&=\sigma_n(\I+\F^\star\A^\star)^{\T}\mathbf d^\star,\\
\mathbf b_{\mathrm{min}}
&=\sigma_z{\F^\star}^{\T}\mathbf d^\star.
\end{aligned}
\]
Their stacked difference lies in the nullspace of the
constraint matrix in \eqref{eq:variational_rep_SNR}, whereas
the stacked minimum-norm solution lies in its row space.
Thus, by orthogonality,
\[
\begin{aligned}
&\|\mathbf a_{\mathrm{min}}-\mathbf a_0\|^2
+\|\mathbf b_{\mathrm{min}}-\mathbf b_0\|^2
\\
&\quad=
\|\mathbf a_0\|^2+\|\mathbf b_0\|^2
-\|\mathbf a_{\mathrm{min}}\|^2-\|\mathbf b_{\mathrm{min}}\|^2
\\
&\quad=
\frac{\|\g^\star\|^2}{\sigma_n^2}
+\frac{\|\A^\star\g^\star\|^2}{\sigma_z^2}
-\SNR_{\mathrm{Alg.1}}.
\end{aligned}
\]
In particular, the second-block difference is
\[
\boxed{
\mathbf b_{\mathrm{min}}-\mathbf b_0
=\frac{\A^\star\g^\star}{\sigma_z}
+\sigma_z{\F^\star}^{\T}\mathbf d^\star
}.
\]
Retaining only its squared norm and applying the power
constraints yields
\begin{align}
&\left\|
\frac{\A^\star\g^\star}{\sigma_z}
+\sigma_z{\F^\star}^{\T}\mathbf d^\star
\right\|^2
\nonumber\\
&\quad\le
\frac{\|\g^\star\|^2}{\sigma_n^2}
+\frac{\|\A^\star\g^\star\|^2}{\sigma_z^2}
-\SNR_{\mathrm{Alg.1}}
\nonumber\\
&\quad\le
\Gap(T)
=
O(1).
\label{eq:appG:feedback-residual-bound}
\end{align}

The first inequality follows from the orthogonal decomposition
relative to the minimum-norm solution in
\eqref{eq:variational_rep_SNR}.
Together with \eqref{eq:appG:message-power-asymptotics},
this yields
\begin{equation}
\sigma_z^2
\|{\F^\star}^{\T}\mathbf d^\star\|^2
=
\frac{T\Prx}{\sigma_z^2}
+O(\sqrt T).
\label{eq:appG:feedback-noise-energy}
\end{equation}

We now use projected stationarity along the causal scaling
\[
(\F,\A)
=
(s^{-1}\F^\star,s\A^\star).
\]
Let $\mathcal L(s)$ denote the Lagrangian along this scaling,
with $\g^\star$ and $(\lambda_1^\star,\lambda_2^\star)$ held fixed.
Since the scaling preserves causality, projected stationarity implies
\[
\left.\frac{d\mathcal L(s)}{ds}\right|_{s=1}=0.
\]
Using the invariance of $\F\A$ along this scaling,
evaluating the derivative and applying complementary slackness yields
\begin{align}
\lambda_2^\star T\Prx
&=
\sigma_z^2
\|{\F^\star}^{\T}\mathbf d^\star\|^2
\nonumber\\
&\quad+
\lambda_1^\star\sigma_z^2\|\F^\star\|_F^2
\nonumber\\
&\quad+
\lambda_2^\star\sigma_z^2
\|\A^\star\F^\star\|_F^2.
\label{eq:appG:scaling-stationarity}
\end{align}
Here
\[
\sigma_z^2\|\F^\star\|_F^2
\le
\Delta_{\rm fw}
=
O(1),
\]
and
\[
\sigma_z^2\|\A^\star\F^\star\|_F^2
\le
\tr\!\left(
\A^\star\Cov_w{\A^\star}^{\T}
\right)
\le
\Delta_{\rm fb}
=
O(1).
\]
Therefore, \eqref{eq:appG:feedback-noise-energy} and
\eqref{eq:appG:scaling-stationarity} imply
\[
\lambda_2^\star T\Prx
=
\frac{T\Prx}{\sigma_z^2}
+O(\sqrt T),
\]
so $\lambda_2^\star\to1/\sigma_z^2$.

Substituting this result into
\eqref{eq:appG:snr-multiplier-identity}, and using
\eqref{eq:appG:message-power-asymptotics} together with
$\Gap(T)=O(1)$, gives
$\lambda_1^\star\to1/\sigma_n^2$.
Consequently,
\begin{equation}
\boxed{
\lambda_1^\star \longrightarrow \frac{1}{\sigma_n^2},
\qquad
\lambda_2^\star \longrightarrow \frac{1}{\sigma_z^2}
\quad \text{as } T\to\infty
},
\end{equation}
which completes the proof. \qed

\end{document}